\documentclass{amsproc}

\usepackage{amsmath,amssymb,amscd,verbatim,graphicx,color}
\usepackage{apacite}
\usepackage[english]{babel}
\usepackage[T1]{fontenc}
\usepackage[multiple]{footmisc}
\usepackage[all]{xy}
\SelectTips{cm}{}

\newcommand\too{\longrightarrow}
\newcommand\inv{^{-1}}

\DeclareMathOperator\Lie{Lie}
\DeclareMathOperator\diag{diag}

\DeclareMathOperator\id{id}

\newcommand{\Hol}{^\mathrm{hol}}

\newenvironment{softthm}{\begin{quote}\hspace{-0.5cm}}{\end{quote}}
\newenvironment{softqt}%
{\begin{list}{}%
	{\setlength{\leftmargin}{1cm}}
	\item[]%
}
{\end{list}}

\theoremstyle{plain}
\newtheorem{thm}{Theorem}

\newtheorem*{question'}{Question \arabic{question}$'$}
\newtheorem{lemma}{Lemma}[section]

\theoremstyle{definition}
\newtheorem{definition}{Definition}
\newtheorem{example}{Example}

\theoremstyle{remark}
\newtheorem*{rem}{Remark}

\newtheorem{remnum}{Remark}[section]

\newcommand\Aa{\mathcal{A}}  \newcommand\CC{\mathcal{C}} \newcommand\DD{\mathcal{D}}  \newcommand\FF{\mathcal{F}}  \newcommand\HH{\mathcal{H}} \newcommand\II{\mathcal{I}}  \newcommand\KK{\mathcal{K}}                  \newcommand\CCC{\mathbb{C}}     \newcommand\HHH{\mathbb{H}}           \newcommand\RRR{\mathbb{R}}        \newcommand\ZZZ{\mathbb{Z}}

\DeclareMathOperator\SL{SL}
\DeclareMathOperator\GL{GL}
\DeclareMathOperator\Mat{Mat}

\DeclareMathOperator\Hom{Hom}

\DeclareMathOperator\Cinf{C^\infty}
\DeclareMathOperator\Cliff{\CC\!\ell}

\newcommand\Lp{\ensuremath{L_+}}
\newcommand\Lpc{\ensuremath{L_+(\CCC)}}
\newcommand\Lpo{\ensuremath{L_+^\uparrow}}
\newcommand\Lpoc{\ensuremath{(L_+^\uparrow)^\CCC}}
\newcommand\Lpot{{\ensuremath{\tilde L_+^\uparrow}}}
\newcommand\Lpotc{\ensuremath{(\tilde L_+^\uparrow)^\CCC}}
\newcommand\Lpt{\ensuremath{{}^b\!\tilde L_+}}
\newcommand\Lpti{\ensuremath{{}^a\!\tilde L_+}}
\newcommand\Lptt{\ensuremath{\tilde L_+}}
\newcommand\Lpa{L_+^{\downarrow}}
\newcommand\Lia{\ensuremath{L_-^\downarrow}}
\newcommand\Lpat{\ensuremath{{}^b\!\tilde L_+^{\downarrow}}}
\newcommand\Lpati{\ensuremath{{}^a\!\tilde L_+^\downarrow}}
\newcommand\Lio{\ensuremath{L_-^\uparrow}}

\newcommand\sltc{\SL(2,\CCC)}

\newcommand\R{\ensuremath{\rho}}
\newcommand\RC{\R^\CCC}

\newcommand\RbF{\ensuremath{[\rho\omega]}}
\newcommand\RhbFc{\ensuremath{[\rho_\#\omega]}}
\newcommand\RbFc{\ensuremath{\RbF}}
\newcommand\RbFs{\ensuremath{\RbF}}

\newcommand\RbFH{\ensuremath{\RbF\Hol}}

\newcommand\RHbPF{\ensuremath{[\rho\Hol\omega']}}

\newcommand\RP{\R'}

\newcommand\RPbPF{\ensuremath{[\rho'\omega']}}

\newcommand\Kf{\KK^\mathrm{form}}
\newcommand\Kt{\KK^\mathrm{test}}
\newcommand\KfC{\KK^\mathrm{form}_+}
\newcommand\Kfc{\KK^\mathrm{form}_c}
\newcommand\Kfs{\KK^\mathrm{form}_s}
\newcommand\Df{\DD^{\mathrm{form}}}
\newcommand\Dd{\DD^{\mathrm{diff}}}

\newcommand\Fc{\ensuremath{\FF_c}} 
\newcommand\Fs{\ensuremath{\FF_{sc}}} 

\newcommand{\bd}{\begin{description}}
\newcommand{\ed}{\end{description}}
\newcommand{\bea}{\begin{eqnarray}}
\newcommand{\eea}{\end{eqnarray}}
\newcommand{\beal}{\begin{equation}\begin{aligned}}
\newcommand{\eeal}{\end{aligned}\end{equation}}
\newcommand{\beals}{\begin{equation*}\begin{aligned}}
\newcommand{\eeals}{\end{aligned}\end{equation*}}
\newcommand{\be}{\begin{equation}}
\newcommand{\ee}{\end{equation}}
\newcommand{\bdm}{\begin{displaymath}}
\newcommand{\edm}{\end{displaymath}}
\newcommand{\mc}[1]{\ensuremath{\mathcal{#1}}}
\newcommand{\bi}{\begin{itemize}}
\newcommand{\ei}{\end{itemize}}
\newcommand{\s}{\,\!}
\newcommand{\ben}{\begin{enumerate}}
\newcommand{\een}{\end{enumerate}}

\title{The CPT Theorem}

\author{Hilary Greaves and Teruji Thomas}

\begin{document}
\maketitle

\section{Introduction and Motivation}
\label{intro}

The CPT theorem says, roughly, that every relativistic quantum field theory has a symmetry that simultaneously reverses charge (C), reverses the orientation of space (or `parity,' P), and reverses the direction of time (T). In this paper we will state and prove a general version of this theorem, proceeding from first principles and explicitly setting out all required assumptions. 

Why re-examine a result that is so widely known?  The motivation stems from the fact that, as a general rule, the QFT literature splits rather sharply into two sectors. The first sector deals with `Lagrangian QFT'; it speaks the language of mainstream particle physics, but is often rather relaxed about mathematical rigour. The second sector is fully rigorous, but bears a much looser relationship to the QFTs that actually enjoy predictive success; it includes the axiomatic program of Streater and Wightman, and the purely algebraic approach (AQFT) associated with e.g. Araki, Haag and Kastler.

The literature on the CPT theorem is no exception to this general rule. In standard Lagrangian-QFT textbooks (e.g. \cite{psQFT}, \cite{izQFT}, \cite{weiQFT1}) the `theorem' is that Lagrangians of a certain kind are necessarily invariant under a CPT transformation of the fields; they establish this result via case-by-case calculations for the fields of most physical interest (e.g. vectors or Dirac spinors in 3+1 spacetime dimensions), and refer the reader to e.g. \cite{swPCT} for a more rigorous and general proof. If one follows up these references, one indeed finds a fully rigorous proof of a result called `CPT Theorem,' but the relationship of \emph{that} result to the CPT invariance of Lagrangians is obscure; the same remark applies to such AQFT results as that presented in \cite{byCPT}.
The literature contains a gap: there is no rigorous, general proof available of the CPT theorem within the framework of \emph{Lagrangian} QFT.

This paper aims to fill that gap. We present a rigorous proof using only the basic geometric and group-theoretic facts on which the CPT result essentially depends. Our approach has the following features.

(1)   We are concerned solely with the symmetries of Lagrangian densities, dynamical equations, and similar objects; we say only enough about quantum field theory \emph{per se} to motivate appropriate transformation laws.  
In fact, our results apply formally to classical relativistic field theories just as well as to quantum ones. We find that the quantum CPT theorem is an instance of a more general result, other instances of which can be seen as classical PT, classical CPT and quantum PT theorems. In standard approaches to the CPT theorem, the relationship between quantum and classical  symmetries is left unclear.

(2) We give a general construction of CPT transformations for an arbitrary field, based only on how that field transforms under proper orthochronous Lorentz transformations.   This construction is clearly tied to the requirements of our proof of the CPT theorem, so it is clear \emph{why} an invariance theorem results for these particular transformations. In the existing Lagrangian-QFT literature, the CPT transformations tend to be introduced \emph{ad hoc} and case-by-case.

(3) We rely on a few basic geometric properties of the Lorentz group, so that our results are valid for Minkowski space, and, indeed, for any non-Euclidean signature, in dimension at least three. These properties are absent in dimension two and for Galilean spacetimes (for which we show there is no analogous result). The standard approach relies on a detailed classification of the representations and invariants of the four-dimensional Lorentz group, thus obscuring the basic structure and generality of the result. 

(4)  Our key technique is passage from the real to the complex Lorentz group. This `complexification' is also the key idea used to prove the CPT theorem of axiomatic QFT, but it plays no overt role in standard approaches to the Lagrangian CPT theorem.\footnote{ Complexification \emph{does} play a key role in the treatment of tensors in an illuminating paper by J. S. Bell \citeyear{belTR}; the latter was the original inspiration for the present paper.}

\subsection*{}
We develop our argument pedagogically, treating first the simpler case of fields taking values in true representations of the Lorentz group (i.e. tensor fields), and later generalising to include properly projective representations (spinor fields). The reader interested only in the broad outline of our results can skip sections \ref{pt_tensor}--\ref{SR_spinor}.

The structure of the paper is as follows. Sections \ref{ffts}--\ref{overview} lay the conceptual foundations.  Section \ref{ffts} 
introduces our basic notion of a `formal field theory,' and explains how it can be used to study the symmetries of classical and quantum field theories.  Section \ref{ptcpt} explains the distinction between PT and CPT transformations, and the related idea of charge conjugation.  
Section \ref{overview} uses this framework to give a detailed overview of our results.

Sections \ref{pt_tensor}--\ref{SR_spinor} form the technical heart of the paper. 
Section \ref{pt_tensor} states and proves a `classical PT theorem':
we show that for classical field theories whose dynamical fields take values exclusively in \emph{true} representations of the Lorentz group (thus excluding spinor fields), proper orthochronous Lorentz invariance entails `PT invariance.'  Section \ref{SR_tensor} generalises the result of section \ref{pt_tensor}: we prove a general invariance theorem that has `tensors-only' versions of the classical PT theorem, the quantum CPT theorem, and classical CPT and quantum PT theorems as corollaries. Of these, the classical PT and quantum CPT theorems are the most interesting, because their premisses are widely accepted.  

We next generalise to spinorial field theories. Section \ref{covers} lays out the basic facts concerning covers of the proper Lorentz group. 
Section \ref{holomorphic} explains how the most straightforward attempt to generalise our classical tensorial PT theorem to include spinors fails.
Section \ref{SR_spinor}, building on this instructive failure, further generalises the results of section \ref{SR_tensor} to the spinorial case; this includes the full quantum CPT theorem. 

Section \ref{axioms} examines how our methods apply beyond Minkowski space. We generalise our results to arbitrary non-Euclidean signatures in dimension at least 3. We also point out why our methods fail in various settings where there is provably no analogue of the CPT theorem.  Section \ref{conclusion} is the conclusion.

Some mathematical background is presented in Appendix \ref{maths}, to which the reader should refer as necessary.  Appendix \ref{Clifford} relates our treatment of the covering groups of the Lorentz group to the usual approach in terms of Clifford algebras. Detailed proofs are relegated to Appendix \ref{proofs}.

\section{Field Theories and Their Symmetries}
\label{ffts}

We will state and prove our invariance theorems in a setting of `formal field theories,' in which the objects of study are formal polynomials that can equally well be interpreted as dynamical equations or as defining Lagrangian or Hamiltonian densities for classical or quantum field theories. The advantage of this framework (over, say, one that takes the objects of study to be spaces of kinematically allowed fields and their automorphisms) is its neutrality between classical and quantum field theories, and between various interpretations of QFTs (as dynamical constraints on operator-valued distributions, formal algorithms for the generation of transition amplitudes, or anything else). 

In this section we explain in detail what a formal field theory is, and how they can be used to describe classical and quantum field theories. In particular, we explain how to analyse \emph{space-time symmetries} of classical and quantum field theories in terms of an analogous notion for formal field theories. 

Initially, `spacetime' $M$ can be any vector space.\footnote{As a matter of convenience, we choose an origin for $M$ (thus making it a \emph{vector} space instead of an affine space). When we discuss symmetries, this choice allows us to focus on the Lorentz group rather than the full Poincar\'e group; it is justified by an implicit assumption that our field theories are, in an appropriate sense, translation invariant.} We must eventually suppose that $M$ has enough structure for us to speak of `time-reversing' transformations.

\subsection{Classical field theories.} A \emph{classical field theory} is a set $\mc D \subset \mc K$, where the set $\mc K \equiv C^\infty(M,V)$ of kinematically allowed fields consists of all smooth functions from spacetime to some finite-dimensional real vector space $V$.\footnote{If the theory `contains two or more dynamical fields,' as e.g. electromagnetic theory contains the Maxwell-Faraday tensor field $F_{\alpha\beta}$ and the charge-current density vector field $J^\alpha$, then $V$ will naturally be written as a direct sum of two or more spaces: $V_{EM} := V_F \oplus V_J$. See Example \ref{Example:Maxwell}.} $\mc D$ is the set of dynamically allowed fields.  
We are mainly interested in theories $\DD$ that consist of the solutions to a system of differential equations with constant coefficients -- for brevity, we say that $\DD$ is \emph{polynomial}, because these field equations depend polynomially on the field components and their derivatives.  

We will allow our differential equations to have complex coefficients. This requires some comment. If we were only interested in classical field theories, it would suffice to consider differential equations with \emph{real} coefficients. By way of example, it is true that the Dirac equation
\be\label{DiracEq}-i\gamma^\mu\partial_\mu\psi + m\psi=0\ee
has complex coefficients; however, by taking real and imaginary parts, we may consider this as a system of two differential equations with real coefficients. As the example also shows, however, it is nevertheless convenient to allow for complex coefficients, of which real coefficients are a special case. More importantly, the use of complex coefficients will be crucial for the study of symmetries in \emph{quantum} field theory.  There  the complex structure of the coefficients can be identified with the complex structure of Hilbert space, but must be sharply distinguished from any complex structure that $V$ may happen to possess (e.g. the way in which a complex scalar field or a Dirac spinor is complex). The latter structure is fundamentally irrelevant to our purposes (cf. Example \ref{complexKG}). 

We now spell out the notion of a polynomial classical field theory more precisely.  
First, let $W=\Hom(V,\CCC)$ be the space of real-linear maps $V\to\CCC$. Given $\Phi\in\KK$, by a \emph{derived component} of $\Phi$ we mean one of the functions
\be\label{component}\Phi^\lambda_{\xi_1\cdots\xi_n}:=\partial_{\xi_1}\cdots\partial_{\xi_n}(\lambda\circ\Phi)\in \Cinf(M,\CCC)\ee
specified by the data of  $\lambda\in W$ and a (possibly empty) list of vectors $\xi_1,\ldots,\xi_n\in M$. A \emph{differential operator} (with constant complex coefficients) is a map $\KK\to\Cinf(M,\CCC)$ that assigns to every $\Phi\in\KK$ a fixed polynomial combination of its derived components -- that is, a finite sum of finite products of them, along with complex scalars. We say that a classical field theory $\DD\subset\KK$ is \emph{polynomial} if there is a set $\Dd$ of differential operators such that
\be\label{polynomial}\Phi\in\DD \iff [D(\Phi)=0\mbox{ for all }D\in\Dd].\ee
The vast majority of classical field theories considered in physics are polynomial in this sense.

\begin{example}\label{Example:Maxwell}
Consider the Maxwell equation usually written (with implicit summation over $\beta$) as 
\be\label{eq:Maxwell}
F^{\alpha \beta}\s_{,\beta} - J^\alpha = 0.
\ee
To illustrate our notation, let $V_F\subset M\otimes M$ be the space of contravariant\footnote{Since the Maxwell-Faraday tensor is the exterior derivative of a one-form,  it is of course most fundamentally a \emph{co}variant anti-symmetric rank two tensor. We ignore this nicety for simplicity of exposition; the background Minkowski metric allows us to raise and lower indices at will.}, skew-symmetric rank-two tensors at a point, and  $V_J=M$ the space of vectors. A pair consisting of a particular Maxwell-Faraday tensor field $F \in C^\infty (M,V_F)$ and a particular charge-current density vector field $J \in C^\infty (M, V_J)$ can then be seen as a single field $\Phi \equiv F \oplus J\in \Cinf(M,V_F \oplus V_J)$. 
With respect to an orthonormal basis $e_0,e_1,e_2,e_3\in M^*$ of covectors, we can rewrite \eqref{eq:Maxwell} as
\be\label{eq:Maxwell2}
\Phi^{e_\alpha e_\beta\oplus 0}_{e^\beta} - \Phi^{0 \oplus e_\alpha} = 0. 
\ee
For each $\alpha\in\{0,1,2,3\}$, the left-hand side of \eqref{eq:Maxwell2} is a differential operator applied to the field $\Phi=F\oplus J$; the set $\Dd$ of these four operators specifies the dynamics of Maxwell field theory, which is therefore a polynomial field theory. 
\end{example}

\begin{example}\label{Sine-Gordon}  Here are two standard examples of non-polynomial field theories.  First, consider the Sine-Gordon equation for a scalar field $\phi$: 
$$\partial_\mu\partial^\mu\phi+\sin\phi=0.$$
Since sine is not a polynomial function, this does not define a polynomial field theory. However, our results could be extended (or applied indirectly) to the Sine-Gordon equation and similar cases in which the field equations involve power series (e.g. the Taylor series of sine) rather than polynomials.

A second type of example is a `non-linear $\sigma$ model,' in which the target space $V$ is not even a vector space, but a manifold.  If $V$ is an algebraic variety, then there is still a notion of `polynomial field theory,' and it should be possible to extend our results in at least some cases. However, in this paper we will only consider the most important case of polynomial field theories with a linear target space.  
 
\end{example}

\newcommand\uPhi{\underline{\Phi}}
\subsection{Formal field theories.}  
We now shift attention from differential operators to the formulae that define them. This abstraction will allow us to treat classical and quantum field theories on the same footing.

A \emph{differential formula} is a polynomial combination of the derived components of a purely symbolic field $\uPhi$. We call these derived components \emph{field symbols}. 
A differential formula $F$ determines a differential operator $D_F$ that assigns to each classical field $\Phi\in\KK$ the same polynomial combination of \emph{its} derived components.
 
Let $\Kf$ be the set of all differential formulae. 
To be quite precise, we understand each field symbol  $\uPhi^\lambda_{\xi_1\cdots\xi_n}$
 as an element $\lambda\otimes(\xi_1\cdots\xi_n)$ of the complex vector space $W\otimes_\RRR TM$, where $TM$ is the tensor algebra of $M$. Then we formally define $\Kf$ to be the \emph{free algebra} $\Kf=\FF(W\otimes_\RRR TM)$ (see Appendix \ref{tensor_algebras}--\ref{free_algebras}).

Our basic objects of study are certain nice sets of differential formulae:

\begin{definition} A \emph{formal field theory} is a complex affine subspace $\Df\subset\Kf$ (see Appendix \ref{app:affine}). 
 \end{definition}

Thus a formal field theory $\Df$ defines a polynomial classical field theory $\DD$ via
\be\label{polynomial2}\Phi\in\DD \iff [D_F(\Phi)=0\mbox{ for all }F\in\Df].\ee
Conversely, given a polynomial classical field theory $\DD$, we obtain a formal field theory $\Df$ as the \emph{largest} collection of differential formulae $F$ satisfying \eqref{polynomial2}. In this case, $\Df$ is actually a complex subspace of $\Kf$.

But this is only one way of interpreting formal field theories. We have so far noted that a single differential formula $F$ determines a dynamical equation $D_F(\Phi)=0$; but we could instead consider $D_F$ as a Lagrangian or Hamiltonian density, from which dynamical equations are to be derived. In this case, we can take $\Df$ to be the set of all differential formulae defining the same density $\II$. This is not a complex subspace of $\Kf$,  since it does not contain zero (unless $\II=0$);  but it is still a complex 
\emph{affine} subspace.

Moving beyond classical field theories, an important feature of our definition is that $\Kf$ is a \emph{non-commutative} algebra. For example, given $\lambda,\mu\in W$, the products $\uPhi^\lambda\uPhi^\mu$ and $\uPhi^\mu\uPhi^\lambda$ are generally different elements of $\Kf$ -- different formulae -- even though 
$\Phi^\lambda\Phi^\mu=\Phi^\mu\Phi^\lambda$ for any $\Phi\in\KK$.
By maintaining this distinction, we leave open the possibility of taking $\uPhi$ to represent a \emph{quantum} field, whose components do not generally commute.\footnote{It is also possible to make sense of non-commutative \emph{classical} fields in various ways -- see our discussion of  supercommutativity in section \ref{SR_spinor}.}

This is exactly what is done in standard approaches to QFT, where the `theory' is specified by a density $\II$, presented, as in the classical case, by a differential formula.   How exactly $\II$ is interpreted may depend on whether we are interested in canonical quantization, path integrals, or other methods; but these questions are largely irrelevant insofar as we can focus not on $\II$ itself, but on the collection $\Df$ of all differential formulae that define it. (We make these comments more precise in section \ref{classical_and_quantum}.)

Thus the formal field theory approach is broadly neutral about \emph{what kind} of field theories we wish to study (classical or quantum?) 
and about \emph{how} we wish to study them. (Lagrangians, Hamiltonians, or dynamical equations? Operator distributions or path integrals?).
Because we are interested in \emph{symmetries} of field theories, the only general requirement is that the theory of interest $\DD$ is specified by a complex affine subspace $\Df\subset\Kf$,  in such a way that symmetries of $\DD$ correspond to some appropriate notion of  symmetries for $\Df$.  Our next task is to explain just what the appropriate notions are.

\subsection{Classical spacetime symmetries} 
\label{ss_classical} First let us consider the situation for classical field theories. A permutation of $\KK$ is a \emph{symmetry} of $\DD$ if it leaves $\DD$ invariant. One typically studies \emph{groups} of symmetries: if a group $G$ acts on $\KK$, we can ask whether $G$ acts by symmetries, i.e. whether $\DD$ is $G$-invariant.
   
We are particularly interested in \emph{spacetime symmetries.} This means that the action $u$ of $G$ on $\KK$ is determined by the data of a representation $(\omega,G,M)$ of $G$ on $M$ and a representation $(\rho, G, V)$ of $G$ on $V$ (cf. \ref{representations} on representations). Namely,
\begin{equation}\label{classical-geom}u(g)\Phi=\rho(g)\circ \Phi\circ \omega(g\inv)\qquad\forall g\in G,\Phi\in\KK.
\end{equation}
We summarize this situation by saying that $G$ \emph{acts geometrically} via, and that $u$ is the \emph{geometric action} corresponding to, $\rho$ and $\omega$.\footnote{This characterisation of spacetime symmetries in terms of geometric actions is general enough to include what are normally called `global internal symmetries' -- these come from geometric actions in which $\omega$ is the trivial representation. \label{fn_internal}}

\begin{example}
The basic example is when $G$ is a subgroup of the Lorentz group (or, later, a covering group of such a subgroup); $G$ then acts naturally on $M$, so to get a geometric action, it remains to specify a representation of $G$ on $V$. For example, in Maxwell theory (Example \ref{Example:Maxwell}), we specified that each summand of $V=V_M\oplus V_J$ was a space of contravariant tensors. 
\end{example}

How can we study spacetime symmetries using formal field theories?
In general, if a group $G$ acts on $\Kf$, we can ask whether a formal field theory $\Df$ is $G$-invariant.  
However, we are only interested in actions $\RbFc$ of $G$ on $\Kf$ that correspond appropriately to a geometric action $u$ of $G$ on $\KK$. The natural matching condition is   \be\label{classicalcorr}
D_{F}(u(g)\inv\Phi)=D_{\RbFc(g)F}(\Phi)\circ\omega(g)\quad\mbox{for all $g\in G,F\in\Kf,\Phi\in\KK.$}\ee
This and \eqref{classical-geom} determine $\RbFc$, which can be described more explicitly as follows.

\begin{definition} \label{f_defn}
Let $(\RbFc,G,\Kf)$ be the unique representation satisfying the following properties (cf. \ref{extending_reps}): 
	\ben
	\item For all $g\in G$, $\lambda \in W$, and $\xi_1\cdots\xi_n\in  M,$
		\begin{equation*}
		\RbFc(g) (\uPhi_{\xi_1\cdots\xi_n}^\lambda) = 
\uPhi_ {\omega(g)\xi_1\cdots\omega(g)\xi_n}^{\lambda\circ\rho(g\inv)}. \label{f_b}
		\end{equation*}
	\item Every  $\RbFc(g)$ is an automorphism of algebras: for all $X_1,X_2 \in \Kf$,
$$\begin{aligned} \RbFc(g)(X_1 X_2) &=  \RbFc(g)(X_1) \cdot \RbFc(g)(X_2)\\
\RbFc(g)(X_1+X_2) &=  \RbFc(g)(X_1) + \RbFc(g)(X_2).\end{aligned}$$
	\een 
We call $\RbFc$ the \emph{classical action} of $G$ on $\Kf$ induced by $\rho$ and $\omega$.
\end{definition}
The most important consequence of this definition is that a classical field theory $\DD$ is $G$-invariant with respect to the geometric action $u$ if and only if the largest corresponding  $\Df$ is $G$-invariant with respect to the classical action $\RbFc$ (or, as we normally say,  if and only if $\Df$ is $\RbFc(G)$-invariant). We can therefore analyse classical spacetime symmetries in terms of the  symmetries of formal field theories. 

\subsection{Quantum field theories and spacetime symmetries.} \label{classical_and_quantum}

The situation for quantum field theories is formally the same: a theory is typically specified by a Lagrangian or Hamiltonian density $\II$, which is given by a differential formula. If we take $\Df$ to be the set of all formulae that define the same $\II$, then the invariance of $\Df$ under appropriate transformations of the field symbols corresponds to the existence of spacetime symmetries of the quantum field theory.  

How exactly to interpret the density $\II$, and how to construct a field theory from it, are questions much more difficult in the quantum than in the classical case. They are, at least heuristically, the subject of standard textbooks on quantum field theory. The theorems in this paper focus on the narrow question of the invariance of formal field theories. In doing so, we will also be silent about spontaneous and anomalous symmetry breaking, which disrupt the inference from symmetries of $\II$ to symmetries of the quantum field theory. 

Nonetheless, we must at least say enough to make clear the quantum analogue of Definition \ref{f_defn}. That definition  was determined by the particular relationship between classical and formal field theories, so we cannot rely upon it here. It turns out that the right `quantum action' $\RbF_q$ differs from the classical action $\RbF$  for those $g\in G$ such that $\omega(g)$ is time-reversing. Here, and henceforth, we  assume that spacetime $M$ is equipped with a temporal orientation that is either preserved or reversed by each $\omega(g)$. Let us give the definition, and then some motivating comments.
\begin{definition} \label{f_defnq}
Let $(\RbFc_q,G,\Kf)$ be the unique representation satisfying the following properties: 
	\ben
	\item For all $g\in G$, $\lambda \in W$, and $\xi_1\cdots\xi_n\in  M,$
		\begin{equation*}
\RbFc_q(g) (\uPhi_{\xi_1\cdots\xi_n}^\lambda) = 
\begin{cases} 
\uPhi_ {\omega(g)\xi_1\cdots\omega(g)\xi_n}^{*\circ \lambda\circ\rho(g\inv)}
& \mbox{if  $\omega(g)$ is time-reversing} \\
\uPhi_ {\omega(g)\xi_1\cdots\omega(g)\xi_n}^{\lambda\circ\rho(g\inv)} 
&\mbox{otherwise.} 
\end{cases}\end{equation*}
(Here $*\colon\CCC\to\CCC$ is complex conjugation.) 
	\item Every  $\RbFc_q(g)$ is an automorphism of algebras: for all $X_1,X_2 \in \Kf$,
$$\begin{aligned} \RbFc_q(g)(X_1 X_2) &=  \RbFc_q(g)(X_1) \cdot \RbFc_q(g)(X_2)\\
 \RbFc_q(g)(X_1+X_2) &=  \RbFc_q(g)(X_1) + \RbFc_q(g)(X_2) .  \end{aligned}$$
	\een 
We call $\RbFc_q$ the \emph{quantum action} of $G$ on $\Kf$ induced by $\rho$ and $\omega$. 
\end{definition}
\noindent
Our general assumption, then, is that for each quantum field theory $\DD$ of interest, there exists a formal field theory $\Df$ such that $\Df$ is $\RbF_q(G)$-invariant if and only if $G$ acts by spacetime symmetries on $\DD$. Our theorems, which are results about formal field theories, will apply to quantum field theories insofar as this assumption holds.

In the remainder of this section, we sketch one story about why this assumption holds, following (and, we hope, clarifying) typical textbook treatments of CPT invariance. In doing so, our aim is solely to provide the reader with a bridge to the literature: we do not claim that the view of quantum field theory offered here is particularly perspicacious, and indeed it is well known that Haag's Theorem severely undermines the `interaction picture' to which we (following the textbooks) eventually appeal.\footnote{See \cite{Earman-Haag} for a discussion.
}
(The reader already happy that $\RbFc_q$ is the appropriate definition can skip to section \ref{ptcpt}.)

What, first of all, is a quantum field theory?  According to the ideal articulated by the Wightman axioms,
\footnote{We omit some features that are unimportant to our present aim. For a complete axiomatization, and a proof of the CPT theorem within this framework, see \cite{swPCT}. } 
a quantum field theory is at heart is a triple $(\Kt, \HHH, Q)$, where $\Kt$ is a space of `test functions' $M\to V^*$, $\HHH$ is a Hilbert space, and the `quantization map' $Q$ associates to each $f\in\Kt$ a Hermitian operator $Q(f)$ on $\HHH$. Let $\Aa$ be the space of all Hermitian operators.  The eponymous `quantum field' is a distribution $\Phi$ on $M$ with values in $\Aa\otimes_\RRR V$. It is defined by the property that $Q(f)$ is the integral of $f$ against $\Phi$, contracting $V$ with $V^*$. As in the classical case, we can speak of the \emph{derived components} of $\Phi$, defined by \eqref{component}; but these components are operator-valued distributions on $M$, rather than functions $M\to\CCC$.

A symmetry of $(\Kt,\HHH,Q)$ is naturally defined to be an automorphism of the data, i.e. a pair of maps $(u\colon\Kt\to\Kt,U\colon \HHH\to\HHH)$ such that
$$
U\circ Q(f)\circ U\inv  = Q(u(f)) \quad\mbox{for all $f\in\Kt$}.
$$
$U$ should also preserve some of the structure of $\HHH$: it should map rays to rays, and preserve transition probabilities. According to a theorem of Wigner, this means that 
$U$ is either complex-linear and unitary or else anti-linear and anti-unitary.\footnote{See the Appendix A to chapter 2 in \cite{weiQFT1}.}

As in the classical case, however, our interest is not in arbitrary `symmetries' in this minimal sense, but in those corresponding in a certain way to underlying actions of the same group on $V$ and on $M$. We again start from the notion of a geometric action of $G$ on $\KK$, as defined in section \ref{ss_classical}. In the quantum case, the issue is whether a given geometric action\footnote{Note that elements of $\Kt$ are classical fields with values in $V^*$ rather than $V$ (heuristically, elements of $\Kt$ are classical observables rather than classical fields). But if $G$ acts geometrically on $\KK$ via $\omega$ and $\rho$, then it also acts geometrically on $\Kt$ via $\omega$ and the dual representation $\rho^*$. }
$u$ of $G$ on $\Kt$ extends to an action of $G$ by symmetries of $(\Kt,\HHH,Q)$: that is, whether for each $g\in G$ there exists a transformation $U(g)$ of $\HHH$ such that
\be\label{quantumst} 
U(g)\circ Q(f)\circ U(g)\inv = Q(u(g)f)
\quad\mbox{for all $f\in\Kt$.}\ee
If so, we can say that $G$ acts by spacetime symmetries on the quantum field theory. Remember that, in principle, each $U(g)$ is allowed to be either complex-linear or anti-linear. However -- and here is the key point --  the requirement of a positive energy spectrum  entails that 
\be\label{al}
\mbox{$U(g)$ \emph{is anti-linear if and only if $\omega(g)$ reverses the direction of time}}\ee
(cf. \cite{weiQFT1}, ch. 2.6). This rule and \eqref{quantumst} completely determine how the derived components of $\Phi$ transform when conjugated by $U(g)$:  
\begin{equation}\label{symfields}
U(g)\circ \Phi_{\xi_1\cdots\xi_n}^\lambda(x)\circ U(g)\inv = 
\begin{cases} 
\Phi_ {\omega(g)\xi_1\cdots\omega(g)\xi_n}^{*\circ \lambda\circ\rho(g\inv)}(\omega(g)x)
& \mbox{if  $\omega(g)$ is time-reversing} \\
\Phi_ {\omega(g)\xi_1\cdots\omega(g)\xi_n}^{\lambda\circ\rho(g\inv)}(\omega(g)x) 
&\mbox{otherwise.} 
\end{cases}\end{equation}
Conversely, if \eqref{symfields} holds for all $g\in G$ and all derived components, then $U$ defines an action of $G$ by spacetime symmetries.  This establishes the salience of Definition \ref{f_defnq}:   suppose that $\II_F(x)$ is a polynomial in the derived components\footnote{There is no simple way to make sense of a `polynomial combination of the derived components' because of the distributional nature of the quantum field. Some regularization must be used. For example, when (as below) the field in question is the free `interaction picture' field, $\II_F$ can be defined by a
 \emph{normal ordered} polynomial.} 
 of $\Phi$, as specified by a differential formula $F\in\Kf$.
Then $\RbFc_q$ is the unique representation of $G$ on $\Kf$ such that 
\begin{equation}\label{sympoly}
U(g)\circ \II_{F}(x)\circ U(g)\inv
=\II_{\RbFc_q(g)F}(\omega(g)x).\end{equation}
This is the quantum analogue of \eqref{classicalcorr}.

Formulas \eqref{quantumst}--\eqref{sympoly} explain what it means for $U(g)$ to be a spacetime symmetry of a given quantum field theory, corresponding to a given geometric action $u$ of $G$.  But nothing we have said so far establishes whether such a symmetry $U(g)$ \emph{exists}.  
In order for our results concerning formal field theories to be relevant to
quantum field theories, we need this existence condition to be equivalent to the $\RbF_q(g)$-invariance of a formal field theory. To establish that it is so
equivalent, textbooks typically turn to the `interaction picture.'
One starts from a well-understood free (`interaction picture') quantum field theory. One constructs the interacting (`Heisenberg picture') theory using an `interaction Hamiltonian density' $\II$, a normal-ordered polynomial in the derived components of the free field. The construction is such that 
if $G$ acts by spacetime symmetries on the free theory, and the density $\II$ transforms as a scalar
\begin{equation}\label{symH}U(g)\circ\II(x)\circ U(g)\inv = \II(\omega(g)x),\end{equation}
then $G$ also acts by spacetime symmetries on the interacting theory. 
 However, given \eqref{sympoly}, \eqref{symH} is equivalent to the $\RbF_q(g)$-invariance of the set $\Df$ of all differential formulae defining $\II$. Thus, modulo the relatively straightforward study of free quantum field theories, the existence of quantum spacetime symmetries can be deduced from the invariance of this formal field theory. 
We will consider the case of CPT symmetries of free theories in section \ref{freesym2}. Of course, other (perhaps more satisfactory) ways of understanding interacting theories may not require any reduction to the free case.

\begin{remnum} In defining $\Df$ we were vague about \emph{which} formulae define the same density $\II$. This will be determined by the way in which the field components commute with one another, and hence relies on the spin-statistics connection.  A secondary consideration is that one may wish to consider Lagrangian densities to be `the same' if they differ only by a total derivative. 
\end{remnum}

\begin{remnum}  
Classical and quantum symmetries are closely related, even when they reverse time. Our perspective in sections \ref{SR_tensor} and \ref{SR_spinor} is that (C)PT theorems for classical and quantum field theories are immediate corollaries of the same more general result -- \emph{strong reflection invariance}. 

\end{remnum}

\section{PT, CPT, and Charge Conjugation}\label{ptcpt}

With our general framework in hand,  we can turn to the main focus of this paper: PT and CPT symmetries. 
Our characterisation of PT and CPT transformations does not presuppose the existence of transformations that separately reverse C, P, or T. However, to round out the picture, we also develop the notion of a \emph{charge conjugation} that relates PT to CPT.

In this section, we focus on Minkowski space $M$ of dimension at least 2, although most of what we say generalises to other spacetimes. Thus $M$ is equipped with an inner product $\eta$ of signature $(-+\cdots +)$ or $(+-\cdots-)$.

\subsection*{PT vs. CPT}
The Lorentz group $L$ consists of all linear isometries of $M$: 
$$L=\{g\in\GL(M)\,\mid\,\eta(gv,gw)=\eta(v,w) \mbox{ for all } v,w\in M\}.$$
$L$ has four connected components: the proper orthochronous Lorentz group $\Lpo$ (those transformations, including the identity, that preserve both spatial parity $P$ and time sense $T$),  the improper orthochronous component $\Lio$ (reversing $P$ only), the improper nonorthochronous component $\Lia$ (reversing $T$ only), and the proper nonorthochronous component $\Lpa$ (reversing both P and T).  

Both PT and CPT symmetries are spacetime symmetries corresponding to \emph{proper nonorthochronous} transformations of $M$, that is, to elements of $\Lpa$.  The nomenclature comes from the particle phenomenology of quantum field theory:  CPT transformations exchange particles and anti-particles (thus also reversing the \emph{charge}, C), while PT transformations do not.\footnote{In the literature on  CPT symmetry in four dimensions, it is very common to focus on the single element of $\Lpa$ given by `total reflection' $x\mapsto -x$. But note that
in odd spacetime dimensions, total reflection lies in $\Lia$ rather than $\Lpa$, and so has nothing to do with CPT.}

Although there are no particles in our framework,  we can nonetheless draw the appropriate formal distinction between PT and CPT.
To do this we need an additional datum:  a decomposition $W=W^{+}\oplus W^0\oplus W^{-}$ into complex subspaces, such that complex conjugation $\lambda\mapsto\lambda^*=*\circ\lambda$ interchanges $W^+$ and $W^-$ and fixes $W^0$.\footnote{Recall that $W:= \Hom (V, \CCC)$, where $V$ is the fields' target space.}  We call $W^{+}\oplus W^{0}$ the \emph{particle sector} and $W^{-1}\oplus W^0$ the \emph{anti-particle sector}. Thus $W^0$ corresponds to `neutral particles that are their own anti-particles.'

\begin{remnum}\label{decomposing} In practice this decomposition arises in the following way.
Suppose first that $V$ is given as a complex vector space. Then $W$ splits as $W=W^+\oplus W^{-}$, where $W^+$ consists of complex-linear maps, and $W^{-}$ of anti-linear maps. Second, if  $V$ is merely real, we define $W^0=W$. In general, $V$ is given as the direct sum of a complex and a merely real vector space, and therefore $W=W^{+}\oplus W^{-}\oplus W^0$. For motivation and a more detailed version of much the same story, see \cite{walAnt}. The question of whether $V$ counts as complex or merely real is tied to the existence of internal $U(1)$ symmetries.
\end{remnum}

\begin{definition}                            \label{def:conjugating}
We say that a real-linear map $\sigma\colon W\to W$ is \emph{charge-preserving} if $\sigma(W^\epsilon)=W^{\epsilon}$, and \emph{charge-conjugating} if 
$\sigma(W^\epsilon)=W^{-\epsilon}$, for every $\epsilon\in\{+,0,-\}$.  
\end{definition}

Note that if $W=W^0$,  then $\sigma$ may count as both charge-preserving and charge-conjugating, and, in general, $\sigma$ may be \emph{neither}.

Now suppose that $G$ acts geometrically.  Let $\sigma$ denote either the quantum action $\sigma=\RbFc_q$ or the classical action $\sigma=\RbFc$ of $G$ on $\Kf$. Either way, each $\sigma(g)$ preserves $W\subset\Kf$, and so may be charge-preserving or charge-conjugating (or neither).

\begin{definition}                            \label{def:PTvsCPT}
For any $g\in G$ with $\omega(g)\in\Lpa$,  $\sigma(g)$ is a \emph{PT transformation} if it is charge-preserving and it is a \emph{CPT transformation} if it is charge-conjugating.
\end{definition}

\begin{rem}
It is somewhat arbitrary how (and indeed whether) we choose to extend the PT/CPT distinction from quantum to classical field theories, since the state space of a \emph{classical} field theory does not decompose into particle and anti-particle sectors. 
One fairly natural stipulation would be that $\RbFc(g)$ is a CPT transformation if and only if $\RbFc_q(g)$ is a CPT transformation.  We choose instead to insist on Definition \ref{def:PTvsCPT}, which turns out to have the opposite effect: by our convention, 
$\RbFc(g)$ is a CPT transformation if and only if $\RbFc_q(g)$ is a PT transformation.  However, nothing beyond terminological convenience hangs on this choice. 
\end{rem}

\subsection*{Charge conjugation} We now define a general form of automorphism that will play a key role in the interpretation of our general theorems (i.e. Theorems \ref{thm_cpt_tensor} and \ref{thm_cpt_spinor}), and of which charge conjugation in the usual sense is a special case.

For the general construction, let $\$$ be an \emph{involution} of $W=\Hom(V,\CCC)$,  that is, a real-linear map $\$\colon W \rightarrow W$ such that $\$ \circ \$ = \id$. 
Define 
\be\label{Cact}C_\$(\uPhi^\lambda_{\xi_1\cdots\xi_n})= \uPhi^{\$(\lambda)}_{\xi_1\cdots\xi_n}\ee
and extend this to an automorphism of $\Kf$ by the rules
$$C_\$(XY)=C_\$(X)C_\$(Y)\qquad
C_\$(X+Y)=C_\$(X)+C_\$(Y)$$
for all $X,Y\in\Kf$. 
Assuming that $\$$ is either complex-linear ($\$(i \lambda)= i \$(\lambda)$) or anti-linear ($\$(i \lambda)= - i\$( \lambda)$),  this defines a unique complex-linear or anti-linear automorphism of $\Kf$, which we call \emph{$\$$-conjugation}.  There are two main cases of interest.

First, by an \emph{internal charge conjugation} we mean an involution $\#\colon V\to V$  such that $\lambda\mapsto \#(\lambda):=\lambda\circ\#$ is a charge-conjugating transformation of $W$. $C_\#$ is the type of `charge conjugation' standard in QFT. We claim that if $\sigma$ is a classical or quantum PT transformation, then $C_\#\circ\sigma$ is a similarly classical or quantum CPT transformation.  Indeed,  we can use $\#$ to define a new representation\footnote{Strictly speaking, for $\rho_\#$ to be a representation, we must assume that $\#$ commutes with every $\rho(g)$.} $\rho_\#$ of $G$ on $V$, by 
$$\rho_\#(g)=\begin{cases} \#\circ\rho(g) & \mbox{if $\omega(g)$ reverses time}\\
 \phantom{\#\circ{}}\rho(g) & \mbox{if $\omega(g)$ preserves it.}
\end{cases}
$$
Then, for example, if $\omega(g)$ is time-reversing, 
$C_\#\circ\RbFc(g)$ is just the classical action $\RhbFc(g)$, and it is clear that if $\RbFc(g)$ is charge-preserving then $\RhbFc(g)$ is charge-conjugating, and vice versa. 

It is also interesting to consider $\$=*$, i.e. $\$(\lambda)(v)=\lambda(v)^*$. Then the quantum and classical actions of a group $G$ are related by $C_*$:
$$
\RbFc_q(g)=\begin{cases} C_*\circ \RbFc(g) & \mbox{if $\omega(g)$ reverses time} \\
\phantom{C_*\circ{} } \RbFc(g) & \mbox{if $\omega(g)$ preserves it.}
\end{cases}
$$
Moreover, $C_*$ is always charge-conjugating.  Therefore $C_*$ relates classical PT to quantum CPT,  and classical CPT to quantum PT.

\begin{example} \label{complexKG}
Consider a theory of a `complex scalar field.' This means that the target space $V$ is $\RRR^2$, and that $L$ acts geometrically via the usual action $\omega$ on $M$ and the \emph{trivial} action $\rho$ on $V$. 
Define $\lambda\colon V\to\CCC$ by $\lambda(x,y)=x+iy$.  We divide $W$ into particle and anti-particle sectors by setting $W^+=\CCC\lambda,$ $W^-=\CCC\lambda^*$, and $W^0=0$.  Then we have an internal charge conjugation defined by $\#(x,y)=(x,-y)$.  Indeed $C_\#(\uPhi^\lambda)=\uPhi^{\lambda^*}$.
For any $g\in\Lpa$, one has $\RbF(g)(\uPhi^\lambda)=\uPhi^\lambda$, so in this case $\RbF(g)$ is charge-preserving, so a (classical) PT transformation. Thus $F=i\uPhi^\lambda=\uPhi^{i\lambda}\in\Kf$ transforms as  
$$\begin{aligned}
\RbF(g)(i\uPhi^\lambda)&=i\uPhi^{\lambda\phantom{^*}} 		\\
\phantom{C_\#\circ{}}\RbF_q(g)(i\uPhi^\lambda) & =-i\uPhi^{\lambda^*}	\\ 
C_\#\circ\RbF(g)(i\uPhi^{\lambda})&=i\uPhi^{\lambda^*}		\\
C_\#\circ \RbF_q(g)(i\uPhi^\lambda)&=-i\uPhi^{\lambda}
\end{aligned}$$
under classical PT, quantum CPT, classical CPT, and quantum PT respectively.

One usually says that $V=\CCC$ and that $\#$ is complex conjugation. This is convenient and harmless as long as one carefully distinguishes between the complex structure of $V$ and the complex structure of $W$ and $\Kf$. (This corresponds in QFT to the distinction between the way that fields can be complex and the way that Hilbert space is complex.) For example,
 $C_*$ and $C_\#$ are not equal, even though they are both `complex-conjugation.' Indeed, $C_\#$, as usual for charge-conjugation in QFT, is complex-linear on $\Kf$, while $C_*$ is anti-linear.\end{example}

\begin{rem} We can use an internal charge conjugation $\#\colon V\to V$ to define a geometric action of the group $\ZZZ_2=\{\pm1\}$, acting trivially on $M$. $C_\#$ is just the corresponding classical or quantum action of $\ZZZ_2$ on $\Kf$ (it makes sense in both contexts).  Thus, in our language, charge conjugation can count as a `spacetime symmetry' (cf. footnote \ref{fn_internal}). These comments do not apply to $C_*$, since it does not come from a transformation of $V$.
\end{rem}

\section{PT and CPT Theorems: An Overview}
\label{overview}

We now give an overview of our main results. From now on we assume that $M$ is Minkowski space of dimension at least $3$. We will consider the two-dimensional case and other possible generalisations in section \ref{axioms}.

\subsection{A Classical PT Theorem for Tensors.} \label{geomPTpreview} Initially we are interested in geometric actions of the proper Lorentz group $\Lp=\Lpo\cup\Lpa$. 
Such field theories are called \emph{tensorial}, in contrast to \emph{spinorial} theories in which the Lorentz group is replaced by a covering group.  
When speaking of geometric actions of $\Lp$, we assume in this section  that the action of $\Lp$ on $M$ is the standard one; in terms of differential operators, this means that partial derivatives transform as expected.

Our first result (section \ref{pt_tensor}) has the following form:
\begin{softqt}
\textbf{Classical PT Theorem for Tensors.}
Every geometric action of $\Lpo$ extends, in a certain way, to a geometric action of $\Lp$, such that, with respect to the corresponding \emph{classical} actions on $\Kf$:
\ben\item[(1)] every $\Lpo$-invariant formal field theory is $\Lp$-invariant; 
\item[(2)] if  $\Lpo$ is charge-preserving, then so is $\Lpa$.   
\een
\end{softqt}
In short, the theorem predicts the existence of classical PT symmetries for any $\Lpo$-invariant formal field theory. 
It is obviously \emph{not} true that the invariance predicted in part (1) holds for an arbitrary geometric action of $\Lp$. Rather, our claim is that there exists a \emph{specific} universal way to extend 
geometric actions from $\Lpo$ to $\Lp$,  relative to which $\Lpo$-invariance implies $\Lp$-invariance.

\begin{example}\label{MaxwellPT}
 For the case of Maxwell's equations (Example \ref{Example:Maxwell}), we can observe that (a) the theory is invariant under $\Lpo$ \emph{and} $\Lpa$, if we stipulate that $F$ transforms as a contravariant rank-two tensor, and $J$ as a vector; (b)  the theory is invariant under $\Lpo$ \emph{but not} $\Lpa$ if we stipulate that $F$ transforms as a tensor and $J$ as a pseudo-vector (so that under a total reflection $r\colon x\mapsto -x$ of spacetime, we have $F \mapsto F\circ r$ and $J \mapsto J\circ r$). 
\end{example}

\subsection{A Quantum CPT Theorem for Tensors}\label{cpt_preview}
In section \ref{SR_tensor} we use the above classical PT theorem to derive a result that we call \emph{strong reflection invariance} (see \ref{SRpreview} below). This implies a quantum CPT theorem of the following form.
\begin{softqt}
\textbf{Quantum CPT Theorem for Tensors.}
Every geometric action of $\Lpo$ extends, in the same way as before, to a geometric action of $\Lp$.  With respect to the corresponding \emph{quantum} actions on $\Kf$:
\ben
\item[(1)] every $\Lpo$-invariant formal field theory (satisfying some conditions) is  $\Lp$-invariant; 
\item[(2)] if $\Lpo$ is charge-preserving, then $\Lpa$ is charge-conjugating.
\een
\end{softqt}
The extra conditions in (1) are that the formal field theory is \emph{Hermitian} and \emph{commutative}: the latter amounts to half of the spin-statistics connection, that tensor fields commute (see footnote \ref{spinstatfn} for discussion). Note that these conditions are irrelevant to the preceding classical PT theorem.
 
\subsection{A Quantum CPT Theorem for Spinors.} \label{scpt_preview}
If we were convinced that fields in all theories of interest to physics took values in true representations of $\Lpo$, the above results would suffice to establish the generality of classical PT and quantum CPT invariance.  However, this is not the case: in many examples, the fields take values in \emph{projective} representations of the Lorentz group.%
\footnote{One standard motivation for considering projective representations is that `physical states' in quantum theory correspond to rays, rather than vectors, in a Hilbert space $\HH$. Thus the action of the Lorentz group on the state space amounts to a  projective representation on $\HH$. Such a representation can be constructed by quantizing a classical field theory with values in a finite-dimensional projective representation, of the type we consider here. 

However, it also makes perfect sense to consider classical fields that transform under covering groups of $\Lp$, with or without the quantum-mechanical motivation. Indeed, such spinor fields play an important role in some approaches to general relativity.
} 
We call such field theories `spinorial.' They include the earlier `tensorial' theories as a special case.

Projective representations of $\Lpo$ are the same as \emph{true} representations of a double covering group $\Lpot$ of $\Lpo$.%
\footnote{$\Lpot$ is the universal covering group of $\Lpo$ (section \ref{covering_groups}), except when $\dim M=3$ (see Remark \ref{dim3}). For an arbitrary connected Lie group in place of $\Lpo$, projective representations may not correspond to representations of a covering group -- one must also allow for central extensions of the Lie algebra. See \cite[\S2.7]{weiQFT1}.
}
Thus the assumption of a spinorial (C)PT theorem is invariance under $\Lpot$, and the conclusion should be invariance under, not $\Lp$ itself, but a covering group of $\Lp$ containing $\Lpot$.
We investigate such covering groups in 
section \ref{covers}. It turns out (section \ref{holomorphic}) that the classical PT theorem fails to generalise naively to spinors, and yet strong reflection invariance \emph{does} generalise (section \ref{SR_spinor}). This yields (\textit{inter alia}) a theorem  of the following form.

\begin{softqt}
\textbf{Quantum CPT Theorem for Spinors.}
 There exists a covering group $\Lpt=\Lpot\cup\Lpat$ of $\Lp$,  such that
every geometric action of $\Lpot$ extends, in a certain way, to a geometric action of $\Lpt$. With respect to the corresponding \emph{quantum} actions on $\Kf$: 
\ben
\item[(1)] every $\Lpot$-invariant formal field theory (satisfying some conditions) is $\Lpt$-invariant;
\item[(2)] if $\Lpot$ is charge-preserving, then $\Lpat$ is charge-conjugating.  
\een
\end{softqt}
The conditions required in (1) are that the theory is \emph{Hermitian} and \emph{supercommutative}. The latter is a version of the full spin-statistics connection. 
\subsection{Strong Reflection Invariance}\label{SRpreview}

Our exposition of the quantum CPT theorems in both sections \ref{SR_tensor} and \ref{SR_spinor} proceeds by first establishing a more general invariance theorem, which predicts invariance under what the we call \emph{strong reflections}.\footnote{The idea of strong reflections is prevalent in the early CPT literature (see the discussion in Pauli \cite{pauCPT}, who attributes it to Schwinger). Some authors (e.g. \cite{ticQFT}) argue that strong reflection invariance is just what one should mean by `CPT invariance.' We are not convinced by these arguments, but since we prove both strong reflection invariance and CPT invariance, there is room to disagree.}
A strong reflection is a transformation of $\Kf$ defined by applying a classical PT transform to the field symbols, while reversing the order of products.  
Strong reflection invariance depends on $\Lpo$-invariance and spin-statistics; unlike the CPT theorems, it does not require any Hermiticity assumption. On the other hand, strong reflections cannot be directly interpreted as spacetime symmetries.%

Strong reflection invariance easily implies the quantum CPT and classical PT theorems, as well as quantum PT and classical CPT theorems (with restrictive premisses).  This justifies our earlier remark that classical and quantum invariance theorems are `instances of the same more general result.'

\section{The Classical PT Theorem for Tensor Fields}
\label{pt_tensor}

We now explain in detail the Classical PT Theorem of section \ref{geomPTpreview}.   

\subsection*{Extending Representations} 
We must first show how to extend any geometric action of $\Lpo$ to a geometric action of $\Lp$. This means that, given a representation $(\R,\Lpo,V)$ of $\Lpo$, we must extend it to a representation $(\RP,\Lp,V)$ of all of $\Lp$ on the same space $V$. We do this in such a way that if $\omega$ is the standard representation of $\Lpo$ on $M$, then (letting $V=\omega$ in our construction) $\omega'$ is also the standard representation of $\Lp$ on $M$.  
We proceed in three steps.

\subsection*{Step 1: \it Complexification} Recall (\ref{complexify_vector_spaces}--\ref{complexifying_representations}) 
that any connected Lie group $G$ has a \emph{complexification} $G^\CCC$, which is a complex Lie group, and any representation $(\R,G,V)$ extends canonically to a holomorphic representation $\RC$ of $G^\CCC$ on  $V^\CCC=\CCC\otimes V.$ We can apply this to the case $G=\Lpo$ to obtain a representation $(\RC, \Lpoc, V^\CCC)$.
Thus, to make explicit our requirements,  we have used
\begin{quote}
\textbf{(PT-1)} $\Lpo$ is connected.\end{quote}

\subsection*{Step 2: \it Restriction to $\Lp$} 
Now we want to restrict from a representation of $\Lpoc$ to a representation of $\Lp$. This uses:
\begin{quote}\textbf{(PT-2)} $\Lp$ is a subgroup of the complexification $\Lpoc$ of $\Lpo$.\end{quote}
To prove (PT-2), we identify $\Lpoc$ with something familiar: the \emph{proper complex} Lorentz group. Recall the definition. 
\emph{Complex Minkowski space} is the complex vector space $M^\CCC=\CCC\otimes M$. The inner product $\eta$ on $M$ extends by complex-linearity to a complex-valued inner product $\eta^\CCC$  on $M^\CCC$:
\be\label{etac}
\eta^\CCC(a+bi,c+di):=\eta(a,c)-\eta(b,d)+i[\eta(a,d)+\eta(b,c)].\ee
The \emph{complex Lorentz group} $L(\CCC)\subset\GL_\CCC(M^\CCC)$ consists of those complex-linear maps preserving $\eta^\CCC$:
$$L(\CCC)=\{g\in \GL_\CCC(M^\CCC) \,\mid\,\eta^\CCC(gv,gw)=\eta^\CCC(v,w) \mbox{ for all } v,w\in M^\CCC\}.$$
The \emph{proper} complex Lorentz group $\Lp(\CCC)$ is the identity component of $ L(\CCC)$; it consists of those elements with determinant $+1$.
In particular, $\Lp(\CCC)$ contains $\Lp$ as a subgroup, but it is \emph{connected}, unlike $\Lp$, which has two components. Here is a precise restatement of (PT-2).

\begin{lemma} \label{complexlorentz}\label{feature1} The inclusion $\Lpo\to\Lp(\CCC)$ identifies $\Lp(\CCC)$ with the complexification $\Lpoc$. In particular, $\Lp$ is a subgroup of $\Lpoc$.
\end{lemma}

For the proof of this and other intermediary results, see Appendix \ref{proofs}.

\subsection*{Step 3: \it Restriction to $V$} 
By (PT-2), we can restrict $(\RC,\Lpoc,V^\CCC)$ to a representation of $\Lp$ on $V^\CCC$; but what we want is a  representation of $\Lp$ on $V$. Fortunately, we have the following lemma.
\begin{lemma} \label{lem2}
The transformations $\RC(\Lp)$ preserve $V \subset V^\CCC.$
\end{lemma}

  Now the following definition makes sense.

\begin{definition} Let $(\RP,\Lp,V)$ be the restriction of $(\RC,\Lpoc,V^\CCC)$ to a representation of $\Lp$ on $V$. 
\end{definition}

The proof of Lemma \ref{lem2} in Appendix \ref{proofs} relies on the following more basic fact: 
\begin{quote}\textbf{(PT-3)} Every $g\in\Lp$ is fixed  by complex conjugation of $\Lpoc$.
\end{quote}
The following example establishes (PT-3), as well as the fact that when $\omega$ is the standard representation of $\Lpo$, $\omega'$ is the standard representation of $\Lp$.

\begin{example} \label{eg_vector}
The standard action of $\Lp(\CCC)$ on $M^\CCC$ is holomorphic, so it must be the complexification $\omega^\CCC$ of the standard action $\omega$ of $\Lpo$ on $M$. 
Restricting to $\Lp$, we find that $\omega'$ is just the standard action of $\Lp$ on $M$.  Complex conjugation on $M^\CCC$ is just $*\colon v_1+iv_2\mapsto v_1-iv_2$, for $v_1,v_2\in M$; the $v\in M^\CCC$ fixed by $*$ are just the real vectors $v_1\in M$.  The complex conjugate of $g\in\Lp(\CCC)$ is characterised by the property that $(gv)^*=g^*v^*$, for all $v\in M^\CCC$.  Thus the $g$ fixed by $*$ are those preserving $M\subset M^\CCC$. This of course includes all elements of $\Lp$, whence (PT-3). 
\end{example}

\begin{example}
Suppose, more generally, that  $(\R,\Lpo,V)$ is the tensor representation of type $(m,n)$. That is, $\R$ is the canonical action of $\Lpo$ on $V:=M^{\otimes m}\otimes (M^*)^{\otimes n}$. Then the same sort of argument shows that the representation $(\RP,\Lp,V)$ is just the canonical representation of $\Lp$ on $V$. Compare to Example \ref{MaxwellPT}.

\end{example}

\subsection*{Invariance}

 We are now in a position to state and prove our first fundamental theorem. Suppose that $\Lpo$ acts geometrically via $\rho$ and $\omega$.

\begin{thm}[Classical PT Invariance for Tensors]\label{thm_pt_tensor}
If a  formal field theory is invariant under $\RbF(\Lpo)$, then it is invariant under $\RPbPF(\Lpa)$.
\end{thm}

Of course, the most interesting case is when $\omega$ and hence $\omega'$ are the standard actions of $\Lpo$ and $\Lp$ on $M$.

\begin{proof}
In outline, our proof has two parts. First, the 
classical action $\RbF$ of $\Lpo$ on $\Kf$ extends to a holomorphic representation $\RbFH$ of $\Lpoc$ on the same $\Kf$.  Our first step consists in establishing

\begin{lemma} \label{prop2}
	If $\Df$ is  $\RbF(\Lpo)$-invariant, then it is  $\RbF\Hol(\Lp)$-invariant.
\end{lemma}
This is not yet our goal: we wish to show that $\Df$ is invariant under $\RPbPF(\Lp)$, not  $\RbF\Hol(\Lp)$. However, in fact these two representations are identical; the bulk of our proof consists in establishing this: 
\begin{lemma} \label{prop3}
	$\RbF\Hol = \RPbPF$ as representations of $\Lp$ on $\Kf$.
	\end{lemma}
The proofs of these two lemmas are found in Appendix \ref{proofs}.
\end{proof}

\subsection*{ PT, not CPT}
Suppose we are given a particle/anti-particle decomposition $W=W^+\oplus W^0\oplus W^{-}$, and that the transformations  $\RbF(\Lpo)$ of $W\subset\Kf$ are charge-preserving, i.e preserve this decomposition. For each $\epsilon\in\{+,0,-\}$, we can apply Theorem \ref{thm_pt_tensor} to $\Df=W^\epsilon$. The conclusion is that  $W^\epsilon$ is preserved by $\RPbPF(\Lpa)$; thus $\RPbPF(\Lpa)$ is charge-preserving, so Theorem \ref{thm_pt_tensor} is a PT (not a CPT) theorem.

\section{Strong Reflection and CPT Invariance for Tensors}
\label{SR_tensor}

Theorem \ref{thm_pt_tensor} was relevant only for \emph{classical} field theories, and established only PT (not CPT) invariance. We now turn to the question of \emph{strong reflection invariance}, as previewed in section \ref{SRpreview}.  This implies a range of PT and CPT theorems for both classical and quantum field theories, including especially the quantum CPT theorem for tensors previewed in section \ref{cpt_preview}.

Formally speaking, the results in this section are trivial variants of Theorem \ref{thm_pt_tensor}, but it is these results, and not Theorem \ref{thm_pt_tensor}, that will generalise to the case of spinors.  Stating them independently gives us the opportunity to introduce some fundamental ideas that will find non-trivial application in the general spinorial case.

\subsection*{Commutativity} The basic assumption in this section is that multiplication of field symbols is \emph{commutative}.
This means that we assume identities of the form
$$\uPhi^{\lambda}_{\xi_1\cdots\xi_m}
\uPhi^{\mu}_{\eta_1\cdots\eta_n}=
\uPhi^{\mu}_{\eta_1\cdots\eta_n}
\uPhi^{\lambda}_{\xi_1\cdots\xi_m}.$$
These identities do not hold in $\Kf$. Rather, they are the defining relations of the free commutative algebra $\Kfc=\FF_c(W\otimes_\RRR TM)$ (see \ref{commutativity}). Thus we define a \emph{commutative formal field theory} to be a complex affine subspace $\Df_c\subset\Kfc$.  

However, any commutative formal field theory can also be seen as a formal field theory in the original sense. Indeed, there is a map $c\colon \Kf\to\Kfc$ which identifies two formulae if they differ only by commutation. Instead of talking about a subspace $\Df_c\subset\Kfc$,  we equivalently talk about its inverse image  $\Df=c\inv(\Df_c)\subset\Kf$.  This observation allows us to apply the constructions of section \ref{pt_tensor} to commutative formal field theories.

From the point of view of classical field theory, commutativity is a very natural assumption, because multiplication of the derived components of classical fields is commutative; from the point of view of quantum field theory, it amounts to imposing one half of the `spin-statistics' assumption: that, since we are dealing here exclusively with true (rather than projective) representations of the Lorentz group, all field operators commute with one another (see footnote \ref{spinstatfn} for clarifying discussion). 

\subsection*{Strong Reflection Invariance}

Let $S$ be the transformation of $\Kf$ that is the identity on field symbols --
$$S(\uPhi^{\lambda}_{\xi_1\cdots\xi_n})=\uPhi^{\lambda}_{\xi_1\cdots\xi_n}$$
-- and is an \emph{anti-automorphism} of algebras:
$$S(X+Y)=S(X)+S(Y)\qquad\mbox{but}\qquad S(XY)=S(Y)S(X).$$
A \emph{strong reflection} is a transformation of $\Kf$  of the form  $S\circ\sigma$ for some classical PT (or CPT) transformation $\sigma.$

\begin{thm}[SR Invariance for Tensors]\label{thm_SR_tensor}  If a commutative formal field theory is invariant under $\RbF(\Lpo)$, then it is invariant under $S\circ\RPbPF(\Lpa)$.
\end{thm}

\begin{proof} $S$ is just the identity map on $\Kfc$, since there $XY=YX$. Thus any commutative formal field theory is $S$-invariant. 
According to Theorem \ref{thm_pt_tensor}, $\Df$ is also $\RPbPF(\Lpa)$-invariant, hence invariant under the combination
$S\circ\RPbPF(\Lpa)$.
\end{proof}


Strong reflections, being \emph{anti-automorphisms} of $\Kf$, are not candidates for spacetime symmetries by the lights of Definitions \ref{f_defn} and \ref{f_defnq}. However, we can obtain spacetime symmetries by combining strong reflections with other anti-automorph\-isms of $\Kf$, like Hermitian conjugation.

\subsection*{Hermiticity.} Following the discussion of charge conjugation in section \ref{ptcpt}, let $\$$ be any involution of $W=\Hom(V,\CCC)$,  assumed to be either complex-linear or anti-linear, so that it defines an automorphism $C_\$\colon\Kf\to\Kf$. Define a corresponding anti-automorphism $\dagger_\$$ by
$$\dagger_\$=S\circ C_\$=C_\$\circ S$$ so that $\dagger_\$$ acts the same way \eqref{Cact} as $C_\$$ on field symbols, but reverses the order of products.
We say that any $\Df\subset\Kf$ is \emph{$\$$-Hermitian} if it is invariant under $\dagger_\$$.

\begin{example} The most interesting case is when $\$=*$ is complex conjugation. As we have noted, the complex structure on $W$ corresponds to the complex structure of the quantum Hilbert space; $\dagger_*$ corresponds exactly to Hermitian conjugation of operators in QFT. (Note that Hermitian conjugation, like $\dagger_*$, reverses the order of products.)
\end{example}

\subsection*{PT and CPT theorems} We have

\begin{thm}[General PT/CPT Theorem for Tensors] \label{thm_cpt_tensor}
Suppose that a commutative formal field theory is invariant under $\RbF(\Lpo)$.  Then it is invariant under $C_\$ \circ \RPbPF(\Lpa)$ if and only if it is $\$$-Hermitian.
\end{thm}

\begin{proof} Theorem \ref{thm_SR_tensor} gives invariance under $S\circ\RPbPF(\Lpa)=\dagger_\$\circ (C_\$ \circ \RPbPF(\Lpa))$.  Thus 
$\dagger_\$$-invariance is equivalent to $C_\$ \circ \RPbPF(\Lpa)$-invariance.
\end{proof}

Now we interpret Theorem \ref{thm_cpt_tensor} for different choices of the involution $\$$ (these comments rely heavily on section \ref{ptcpt}). Throughout, $\Df$ is a commutative, $\Lpo$-invariant formal field theory.

\subsection*{\it A Classical PT Theorem.} First take $\$=\id$ to be the identity map.  Then $\dagger_{\id}=S$, and any commutative formal field theory is automatically $\dagger_{\id}$-invariant. Thus Theorem \ref{thm_cpt_tensor} says that
$\Df$ is invariant under the classical PT transformations $\RPbPF(\Lpa)$. This is a slight weakening of Theorem \ref{thm_pt_tensor}, which did not assume commutativity. 

\subsection*{\it A Quantum CPT Theorem. } Second, take $\$=*$ to be complex conjugation.  Theorem \ref{thm_cpt_tensor} says that  $\Df$  is invariant under the quantum CPT transformations $C_*\circ\RPbPF(\Lpa)$ --- equivalently, $\RPbPF_q(\Lpa)$ --- if and only if it is $*$-Hermitian. Since $\dagger_*$ amounts to the usual Hermitian conjugation of operators in QFT, this condition will be met if the interaction density is Hermitian in the usual sense. This is widely assumed to be the case: a $*$-Hermitian Hamiltonian, at least, is needed for unitary time evolution. Thus we get a general prediction of quantum CPT invariance.

\subsection*{\it A Classical CPT Theorem.} Third, take $\$=\#$ to be an internal charge conjugation (see section \ref{ptcpt}). Then  Theorem \ref{thm_cpt_tensor} says that $\Df$ is invariant under the classical CPT transformations $C_\#\circ\RPbPF(\Lpa)$ if and only if it is $\#$-Hermitian. This Hermiticity serves as a convenient criterion for classical CPT invariance,  but there is no general reason it should be met.

\subsection*{\it A Quantum PT Theorem.} Finally, define $\$(\lambda)(v)=\lambda(\#v)^*$ for some internal charge conjugation $\#$; we write `$\$=*\#$.' Theorem \ref{thm_cpt_tensor} now says that $\Df$ is invariant under the quantum PT transformations $C_{*\#}\circ\RPbPF(\Lpa)$ --- equvialently, $C_\# \circ \RPbPF_q(\Lpa)$ --- if and only if it is $*\#$-Hermitian. There is again no general reason this condition should be met. Note this result does not assume that the theory is Hermitian in the usual sense, i.e. $*$-Hermitian. However, when, as usual, $\Df$ is $*$-Hermitian, being $*\#$-Hermitian is equivalent to being $C_\#$-invariant. In other words, we have the usual implication of the CPT theorem, that charge-conjugation invariance is equivalent to PT invariance.
 
\begin{rem}
The commutativity assumption is required for the results in this section, though it plays no role in Theorem \ref{thm_pt_tensor}. We could not instead assume anti-commutativity, because an anti-commutative formal field theory would not be invariant under $S$.  As a trivial example, consider the complex scalar field of Example \ref{complexKG}. The formula
$F=\uPhi^\lambda\uPhi^{\lambda^*}+1$ is $*$-Hermitian and $\Lpo$-invariant, but under quantum CPT transforms to $\uPhi^{\lambda^*}\uPhi^\lambda +1$. If 
$\uPhi^\lambda$ and $\uPhi^{\lambda^*}$ commute, then this is just $F$ again; but if they anti-commute, it equals $-F+2$, and $F$ is not CPT invariant. 
\end{rem}

\section{Covers of the Lorentz Group}
\label{covers}
We now begin to generalise our results to spinors, as explained in section \ref{scpt_preview}. The purpose of this section is to describe the covering groups of $\Lp$, and, in particular, to construct the covering group $\Lpt$ mentioned in our Quantum CPT Theorem. 

We continue to assume that the dimension of Minkowski space $M$ is at least three. At the end we work out an explicit description of all the groups in the four-dimensional case (Example \ref{egSpinRep1}).

\subsection*{Covering groups of $\Lp(\CCC)$.} It is convenient to start our discussion with covering groups of the \emph{complex} proper Lorentz group.  $\Lp(\CCC)$ is connected, but not simply connected. Since it is connected, it has a universal cover $\Lp(\CCC)^\wedge$; since it is not simply connected, $\Lp(\CCC)^\wedge$ is not just equal to $\Lp(\CCC)$. In fact $\Lp(\CCC)^\wedge$ is a double cover. For future reference, it is convenient to state this directly as a property of $\Lpo$, using the fact (Lemma \ref{complexlorentz}) that $\Lp(\CCC)$ is the complexification $\Lpoc$ of $\Lpo$:
\begin{softqt}
\textbf{(PT-4)} The universal cover $\pi\colon (\Lpoc)^\wedge\to\Lpoc$ is a double cover. 
\end{softqt}
We now use this double cover to define a four-fold cover $\pi\colon \Lptt(\CCC)\to \Lp(\CCC)$. 

\begin{definition}
 Let $\{1,\tau\}\subset \Lp(\CCC)^\wedge$ be the preimage of $1\in \Lp(\CCC)$.  Let $\Lptt(\CCC)$ be the group generated by $\Lp(\CCC)^\wedge$ together with a symbol $I$ such that $I^2=\tau$, and such that  $I$ commutes with elements of $\Lp(\CCC)^\wedge$.  Defining $\pi(I)=1$, we obtain a four-fold covering map $\pi\colon\Lptt(\CCC)\to\Lp(\CCC)$.
\end{definition}

The situation is illustrated in Figure 1 (with further details explained below). 

\begin{figure}[h,t]
{
\setlength{\unitlength}{13500sp}%
\begingroup\makeatletter\ifx\SetFigFont\undefined%
\gdef\SetFigFont#1#2#3#4#5{%
  \reset@font\fontsize{#1}{#2pt}%
  \fontfamily{#3}\fontseries{#4}\fontshape{#5}%
  \selectfont}%
\fi\endgroup%
\begin{picture}(722,721)(313,-1109)
\thinlines 
{\put(675,-996){\oval(720,225)[b]}
\put(610,-1085){\makebox(0,0)[lb]{\smash{{\SetFigFont{7}{8.4}{\familydefault}{\mddefault}{\updefault}{$\phantom{I}\Lpc
$}%
}}}}
}
{\put(832,-996){\oval(270,100)[b]}
\put(477,-1030){\makebox(0,0)[lb]{\smash{{\SetFigFont{7}{8.4}{\familydefault}{\mddefault}{\updefault}{$\phantom{I}\Lpo$}%
}}}}
}
{\put(517,-996){\oval(270,100)[b]}
\put(797,-1030){\makebox(0,0)[lb]{\smash{{\SetFigFont{7}{8.4}{\familydefault}{\mddefault}{\updefault}{$\phantom{{}^a}\Lpa$}%
}}}}
}
{\put(675,-900){\vector( 0,-1){81}}}
{\put(675,-771){\oval(720,225)}
\put(610,-860){\makebox(0,0)[lb]{\smash{{\SetFigFont{7}{8.4}{\familydefault}{\mddefault}{\updefault}{$\phantom{I}\Lp(\CCC)^\wedge
$}%
}}}}}
{\put(517,-771){\oval(270,90)}
}
{\put(832,-771){\oval(270,90)}
}
{\put(675,-501){\oval(720,225)}}
{\put(517,-501){\oval(270,90)}}
{\put(832,-501){\oval(270,90)}}
{\put(306,-816){\vector(1, 0){0}}
 \put(306,-726){\vector(1, 0){0}}
 \put(300,-771){\oval(45,90)[l]}
\put(256,-783){\makebox(0,0)[lb]{\smash{{\SetFigFont{7}{8.4}{\familydefault}{\mddefault}{\updefault}{$\tau$}}}}}
}
{\put(306,-546){\vector(1, 0){0}}
 \put(306,-456){\vector(1, 0){0}}
 \put(300,-501){\oval(45,90)[l]}
\put(256,-513){\makebox(0,0)[lb]{\smash{{\SetFigFont{7}{8.4}{\familydefault}{\mddefault}{\updefault}{$\tau$}}}}}
}
{\put(1045,-771){\vector(-1, 0){0}}
\put(1045,-501){\vector(-1, 0){0}}
\put(1055,-636){\oval(45,270)[r]}
\put(1085,-636){\makebox(0,0)[lb]{\smash{{\SetFigFont{7}{8.4}{\familydefault}{\mddefault}{\updefault}{$I\colon\,I^2=\tau$}}}}}
}
{\put(315,-996){\line( 1, 0){720}}
}
\put(687,-938){\makebox(0,0)[lb]{\smash{{\SetFigFont{7}{8.4}{\familydefault}{\mddefault}{\updefault}{$\pi$}%
}}}}
\put(477,-780){\makebox(0,0)[lb]{\smash{{\SetFigFont{7}{8.4}{\familydefault}{\mddefault}{\updefault}{$\phantom{I}\Lpot$}%
}}}}
\put(797,-780){\makebox(0,0)[lb]{\smash{{\SetFigFont{7}{8.4}{\familydefault}{\mddefault}{\updefault}{$\Lpati$}%
}}}}
\put(797,-511){\makebox(0,0)[lb]{\smash{{\SetFigFont{7}{8.4}{\familydefault}{\mddefault}{\updefault}{$\Lpat$}%
}}}}
\put(477,-511){\makebox(0,0)[lb]{\smash{{\SetFigFont{7}{8.4}{\familydefault}{\mddefault}{\updefault}{$I\Lpot$}%
}}}}
\put(610,-591){\makebox(0,0)[lb]{\smash{{\SetFigFont{7}{8.4}{\familydefault}{\mddefault}{\updefault}{$I\Lp(\CCC)^\wedge$}%
}}}}
\end{picture}%
} 
\caption{A four-fold cover $\Lptt(\CCC)$ of the complex proper Lorentz group $\Lp(\CCC)$, with two components. It contains the two double covers of $\Lp$ extending $\Lpot$: $\Lpti=\Lpot\cup\Lpati$ and $\Lpt=\Lpot\cup\Lpat$.  }
\end{figure}

\subsection*{Covering groups of $\Lp$}
Let $\Lpot$ be the preimage of $\Lpo$ in $\Lp(\CCC)^\wedge$. This is a double cover of $\Lpo$; in fact it is the universal cover (except when  $\dim M=3$; see Remark \ref{dim3}).  
We similarly define two different double-covers of $\Lp$, illustrated in Figure 1. First, let $\Lpti=\Lpot\cup\Lpati$, where $\Lpati$ is the preimage of $\Lpa$ in $\Lp(\CCC)^\wedge$. Second, let $\Lpt=\Lpot\cup\Lpat$, where $\Lpat=I\cdot\Lpati$.
It can be shown that any double cover of $\Lp$ containing $\Lpot$ is isomorphic to either $\Lpti$ or $\Lpt$ (we omit the proof).

\begin{example} \label{egSpinRep1}
In four dimensions we have $\Lpot \cong \SL(2,\CCC)$, the group of $2 \times 2$ matrices with complex entries and unit determinant. It is important to bear in mind that, despite notation, $\Lpot$ is only a \emph{real} Lie group; it has no natural complex structure. The covering map $\pi\colon\SL(2,\CCC) \rightarrow \Lpo$ can be specified as follows. Arbitrarily choosing an inertial coordinate system, we can identify $\Lpo$ with a subgroup of $\GL(4,\RRR)$. Hence, to specify a covering map, it suffices to specify an action $\pi$ of $\SL(2,\CCC)$ on $\RRR^4$ preserving the Minkowski norm $x_0^2-x_1^2-x_2^2-x_3^3$. For $x=(x_0,x_1,x_2,x_3) \in \RRR^4$, write
\be \left<x\right> =\left( \begin{matrix} x_0 + x_3 & x_1 - i x_2 \\ x_1 + i x_2 & x_0 - x_3 \end{matrix} \right); \label{x-twiddle}\ee
then, the desired action of $A\in\SL(2,\CCC)$ is given by the matrix multiplication
\be
\left<\pi(A)(x)\right> = A \cdot \left<x\right> \cdot \bar A^T \label{sltc_action}
\ee
(here $\bar A$ is the complex-conjugate of $A$, and $^T$ denotes transpose). The Minkowski norm of $x$ is equal to $\det \left<x\right>$, which is preserved under (\ref{sltc_action}) since $\det A = \det \bar A^T=1$. Note that $\pi$ is two-to-one: $\pi(A)=\pi(-A)$ for all $A \in \sltc.$ 

The universal cover $\Lp(\CCC)^\wedge$ of $\Lp(\CCC)$ is isomorphic to $\sltc \times \sltc$. The covering map is defined 
 as follows. For $x\in M^\CCC=\CCC^4$, define $\left<x\right>$ as in \eqref{x-twiddle}. For $(A,B) \in \sltc \times \sltc$, $\pi(A,B)$ is the linear transformation of $\CCC^4$ given by
\be\label{tranfo}
\left<\pi(A,B)(x)\right> = A \cdot \left<x\right> \cdot B^T.
\ee
Thus $\Lpot$ is identified with the subgroup of pairs $(A,\bar A)$, and 
$\tau$ is represented by the pair $(-1,-1)$ of scalar matrices. 
To describe the four-fold cover $\Lptt(\CCC)$, we represent $I$ by the pair $(i,-i)$ of scalar matrices. This brings us to the following picture, where $H$ is the group of $2\times 2$ complex matrices with determinant $\pm 1$.  
$$\begin{aligned} 
\Lptt(\CCC)&\cong\{(A,B)\in H\times H\mid \det A=\det B \} \\
\Lp(\CCC)^\wedge&\cong\{(A,B)\in H\times H \mid \det A=\det B=1\} \\
\Lpot&\cong\{(A,\bar A)\in H\times H\mid \det A=1\} \\
\Lpati&\cong\{(A,-\bar A)\in H\times H\mid \det A=1\} \\ \Lpat&\cong\{(A,-\bar A)\in H\times H\mid \det A=-1\}.\end{aligned}$$
The covering map $\pi: \Lptt(\CCC) \rightarrow \Lp(\CCC)$ is still given by \eqref{tranfo}. It is four-to-one: for all $(A,B) \in \Lptt(\CCC)$, $\pi(A,B)=\pi(-A,-B)=\pi(iA,-iB)=\pi(-iA,iB)$. 
\end{example}

\begin{remnum}\label{dim3} If $\dim M=3$, then $\Lpot$ is not the universal cover $(\Lpo)^\wedge$ (which turns out to be an infinite cover of $\Lpo$). However,  it is \emph{still} true that any projective representation of $\Lpo$  comes from a representation of $\Lpot$, so there is no loss of generality in considering $\Lpot$ rather than $(\Lpo)^\wedge$. Note that any representation of $\Lpot$ determines a representation of $(\Lpo)^\wedge$, by composing with the covering map $(\Lpo)^\wedge\to \Lpot$. The claim is that every representation of $(\Lpo)^\wedge$ arises in this way.  One can check that
the map $(\Lpo)^\wedge\to \Lpot\subset \Lp(\CCC)^\wedge$ identifies $\Lp(\CCC)^\wedge$ with the complexification of $(\Lpo)^\wedge$ (compare to Lemma \ref{complexlorentzcover}). 
This means that any representation $\rho$ of $(\Lpo)^\wedge$ on $V$ extends to a representation $\rho^\CCC$ of $\Lp(\CCC)^\wedge$ on $V^\CCC$, and therefore $\rho(g)$ depends only on the image of $g$ in $\Lpot$. 
\end{remnum}

\section{A Classical PT Theorem for Spinor Fields?}
\label{holomorphic}

Having described the covering groups of $\Lp$,  we now naively attempt to generalise Theorem \ref{thm_pt_tensor} to the case of spinors. In fact, we will fail in this attempt, but the argument will lead to a generalisation of Theorems \ref{thm_SR_tensor} and \ref{thm_cpt_tensor} in the next section.

\subsection*{} Following the exposition in section \ref{pt_tensor}, we can complexify any representation $(\rho,\Lpot,V)$ to get $(\RC,\Lpotc,V^\CCC)$. Next, we wish to restrict $\RC$ to either $\Lpti$ or $\Lpt$. 
In analogy
 to Lemma \ref{feature1}, we have

\begin{lemma} \label{complexlorentzcover}
The inclusion $\Lpot\to\Lp(\CCC)^\wedge$ identifies $\Lp(\CCC)^\wedge$ with the complexification  $\Lpotc$. In particular, $\Lpti$ is a subgroup of $\Lpotc$.
\end{lemma}

The proof is in Appendix \ref{proofs}.
The result is that we can restrict $\RC$ to a representation of $\Lpti$ (but not of $\Lpt$) on $V^\CCC$. However, this does not mean that $\Lpti$ preserves $V\subset V^\CCC$, and, in fact, the analogue of Lemma \ref{lem2} fails; rather, one has

\begin{lemma} \label{lemRepP'}
Let $(\R,\Lpot,V)$ be a representation of $\Lpot$, and $(\RC, \Lpotc,V^\CCC)$ its complexification. Decompose $V$ as $V=V_0 \oplus V_1$ where $\R(\tau)$ acts as $(-1)^n$ on $V_n$. Then $\RC(\Lpati)$ preserves $V_0$ but maps $V_1$ to $i V_1\subset V^\CCC$.
\end{lemma}

Since $\RC(\Lpati)$ does not preserve all of $V\subset V^\CCC$, there is no obvious way to define a representation of $\Lpti$ on $V$, and therefore no obvious way to associate PT transformations to elements of $\Lpati$. 

\begin{rem} A representation $V_0$ on which $\tau$ acts by the identity is the same thing as a representation of $\Lpo$. Thus we can speak of $V_0$ as the space of `tensors' and $V_1$ as the space of `pure spinors.' When $V=V_0$, $V$ \emph{is} preserved by $\Lpati$, and we \emph{do} get a PT theorem -- namely, Theorem \ref{thm_pt_tensor}.
\end{rem}

\begin{rem}
Just as Lemma \ref{lem2} relied on property (PT-3), so the proof of Lemma \ref{lemRepP'} in Appendix \ref{proofs} reduces to the following fact:
\begin{softthm}
\textbf{(PT-5)} For any $g\in\Lpati$, $g^*= g\tau$.
\end{softthm}
It is automatic from (PT-1)--(PT-4) and the fact that $\Lpa$ is connected, that either (PT-5) holds or else $g^*=g$ for all $g\in\Lpati$. In the latter case, we would have obtained a classical PT theorem in analogy to Theorem \ref{thm_pt_tensor}. 
\end{rem}

\begin{example} \label{eg_spinreps}
We continue Example \ref{egSpinRep1}. Let $\rho$ be the standard representation of  $\Lpot=\sltc$ on $\CCC^2$. The complexity of $V$  is completely irrelevant (cf. Example \ref{complexKG}), so to avoid confusion,  let us write $V=\RRR^4$. For $v=(x,y,z,w)\in V$ define $[v]=(x+iy,z+iw,x-iy,z-iw)\in \CCC^4$. Then $\rho$ is  given by
$$
[\R(A)v]= \left( \begin{matrix} A & 0 \\ 0 & \bar A  \end{matrix}\right)\cdot[v].
$$
Since $\tau$ is represented by $A=-1$, we find $V=V_1$.
The complexification of $V=\RRR^4$ is $V^\CCC=\CCC^4$. 
For $v\in V^\CCC$, define 
$[v]\in\CCC^4$ as before.
Then the complexified representation $\RC$ of $\Lpotc=\Lp(\CCC)^\wedge=\sltc\times\sltc$ on $V^\CCC$ is given by
$$
[\RC(A,B)v] = \left( \begin{matrix} A & 0 \\ 0 & B  \end{matrix}\right)\cdot[v] .
$$
Following Example \ref{egSpinRep1}, $\Lpati$ consists of pairs $(A,-\bar A)$. In particular, one finds that $(1,-1)\in\Lpati$ acts on $V^\CCC$ by 
$$
\RC(1,-1)(x,y,z,w)= (iy,-ix,iw,-iz).
$$
As predicted by Lemma \ref{lemRepP'}, this maps real vectors into purely imaginary ones. 
\end{example}

\subsection*{A Holomorphic Spinorial PT Theorem}
As consolation, there is a class of field theories for which we \emph{can} define a geometric action of $\Lpti$ and prove a PT invariance 
theorem. We will sketch the idea here, but this discussion is merely an aside, and is not used in the rest of the paper. 

Suppose that our representation $(\R,\Lpot,V)$ is complex in the sense that  $V$ is a complex vector space and $\Lpot$ acts complex-linearly. 
This is already enough to define a geometric action of $\Lpti$:  there is a unique extension of $\rho$ to a holomorphic representation $(\R\Hol, \Lpotc, V)$, and, in particular, an action of $\Lpti$ on $V$.  

However, we will only get $\Lpti$-invariance for certain special field theories.
Here is one class of them.  Let $W^+\subset W$ be the subspace of complex-linear maps $V\to\CCC$, and $\KfC=\FF(W^+\otimes_\RRR TM)$ the free complex algebra generated by $W^+\otimes_\RRR TM$. This is a subalgebra of $\Kf$, and it is invariant under the classical action of $\Lpot$.  We call a formal field theory $\Df\subset\Kf$ \emph{holomorphic} if it is 
contained  in $\KfC$. The following theorem is closely parallel to Theorem \ref{thm_pt_tensor}; we omit its proof.

\begin{thm}\label{thm_pt_hol}
If a holomorphic formal field theory is invariant under 
$\RbF(\Lpot)$,  then 
it is invariant under $\RHbPF(\Lpati)$. 
\end{thm}

 However, most theories of physical interest are not `holomorphic' in this sense. If the original representation $\R$ is not complex (e.g. the case of Majorana spinors), then Lemma \ref{complexlorentzcover} tells us that
$\Lpati$  takes fields with values in $V_1$ to fields with values in $iV_1$ -- i.e., outside the original space of kinematically allowed fields. Similarly, if the original representation $\R$ is complex but the dynamics cannot be described by a subspace of $\KfC$, then $\Lpati$ preserves the kinematical space, but not the dynamical subspace. Theorem \ref{thm_pt_hol} therefore does not yield an invariance theorem for field theories that are `non-holomorphic' in either of these two ways.

\begin{example}\label{egDirac}
For the standard theory of the Dirac equation, the space $V$ is a 4-dimensional complex vector space of `Dirac spinors.' The classical field theory given by the Dirac equation \eqref{DiracEq} 
can be modelled by a holomorphic  formal field theory, since the equation depends complex-linearly on the field $\psi$. Corresponding to the total reflection in $\Lp$, there are two elements of $\Lpati$, which yield classical PT transformations $\psi\mapsto \pm\gamma^5\psi=\pm i\gamma^0\gamma^1\gamma^2\gamma^3\psi$. It is easy to check that these are symmetries of the Dirac equation.
On the other hand, consider the constraint
$$\bar\psi\psi=1$$
in the standard Dirac-bilinear notation. This cannot be modelled by a holomorphic formal field theory, because $\bar\psi$ depends anti-linearly on $\psi$.  And in fact our PT transformation $\gamma^5$ is not a symmetry:   $\bar\psi\psi\mapsto -\bar\psi\psi$, whereas $1\mapsto 1$.  (Similarly, Theorem \ref{thm_pt_hol} cannot be applied to the standard Dirac Lagrangian.)
\end{example}

\section{Strong Reflection, PT, and CPT Invariance for Spinors}
\label{SR_spinor}

At last we turn to our main results, previewed in sections \ref{scpt_preview}--\ref{SRpreview}. Throughout we assume that $\Lpot$ acts geometrically via any representation $\rho$ on $V$ and the standard representation $\omega$ on $M$.
\subsection*{Extending Representations}
In the last section we tried and (in general) failed to define classical PT transformations for spinor fields.   Recall that $V$ splits as a direct sum $V=V_0\oplus V_1$, where $\rho(\tau)$ acts by $(-1)^n$ on $V_n$. The problem was that candidate transformations corresponding to elements of $\Lpati$
do not preserve the real space $V$ of field values, instead mapping the space $V_1$ of pure spinors into $iV_1$ (Lemma \ref{lemRepP'}).  
However, this very fact allows us to define a representation of $\Lpt$ (as opposed to: $\Lpti$) that \emph{does} preserve $V$:   

\begin{definition} \label{def2}
Let $(\R, \Lpot, V)$ be any representation of $\Lpot$. Define a representation $(\RP,\Lpt,V)$, extending $\R$, by:
	$$\RP(g)v  =  \begin{cases} \phantom{i^n}\RC(g)v &\mbox{for $g \in \Lpot$} \\ i^n \RC (I^{-1}g)v &\mbox{for $g \in \Lpat$ and $v\in V_n$}. \end{cases} 
$$
\end{definition}
It follows from Example \ref{eg_vector} that if $\rho=\omega$ is the standard representation of $\Lpot$ on $M$, then $\omega'$ is the standard representation of $\Lpt$ on $M$ -- that is, $\omega'$ is  just the projection $\Lpt\to\Lp\subset\GL(M)$.

\begin{example}
We continue Example \ref{eg_spinreps}. Recall (Example \ref{egSpinRep1}) that $\Lpati$ is represented by matrix pairs $(A, -\bar A)$ with $A\in\sltc$,  and $I$ by the pair of scalar matrices $(i,-i)$. Thus $\Lpat=I\cdot \Lpati$ is represented by matrix-pairs $(iA,i\bar A)$.  
Thus the action $\RP$ of $\Lpat$ on $V^\CCC$ is given by 
$$
[\RP(iA,i\bar A)v] = i\left( \begin{matrix} A & 0 \\ 0 & -\bar A  \end{matrix}\right)\cdot[v] .
$$
In particular, one finds that $(i,i)\in\Lpat$, corresponding to a total reflection of $M$,  acts on $V$ by 
$
\RP(i,i)(x,y,z,w)= (-y,x,-w,z).
$
\end{example}

Having shown how to extend geometric actions of $\Lpot$ to $\Lpt$,  we can at least \emph{formulate} analogues of Theorems \ref{thm_pt_tensor} and \ref{thm_SR_tensor}.  However, one cannot expect a direct generalisation of Theorem \ref{thm_pt_tensor} actually to hold, because $\RP$ is not merely  a restriction of the complexification of $\R$. 
 It turns out that we can nonetheless get a direct generalisation of Theorem \ref{thm_SR_tensor},
with the assumption of commutativity replaced by \emph{supercommutativity}, which we now explain.

\subsection*{Supercommutativity.} 
If $V=V_0$, then supercommutativity is just commutativity, as in section \ref{SR_tensor}. If $V=V_1$, we impose instead \emph{anti-commutativity},
$$\uPhi^{\lambda}_{\xi_1\cdots\xi_m}
\uPhi^{\mu}_{\eta_1\cdots\eta_n}=-
\uPhi^{\mu}_{\eta_1\cdots\eta_n}
\uPhi^{\lambda}_{\xi_1\cdots\xi_m}.$$
In general, the decomposition $V=V_0\oplus V_1$ leads to a decomposition 
$W=W_0\oplus W_1$, where $W_n=\Hom(V_n,\CCC)$. 
Then supercommutativity means that 
\be\label{superc}\uPhi^{\lambda}_{\xi_1\cdots\xi_m}
\uPhi^{\mu}_{\eta_1\cdots\eta_n}=
(-1)^{ab}
\uPhi^{\mu}_{\eta_1\cdots\eta_n}
\uPhi^{\lambda}_{\xi_1\cdots\xi_m}\ee
holds for all $\lambda\in W_a$ and $\mu\in W_b$. The relations \eqref{superc} define the
\emph{free supercommutative algebra} $\Kfs=\FF_s(W\otimes_\RRR TM)$ (\ref{free_algebras} and \ref{supercommutativity}).
Thus we can define a \emph{supercommutative formal field theory}
 to be a complex affine subspace $\Df\subset\Kfs$. As with commutative theories, we can consider a supercommutative formal field theory to be a special kind of formal field theory in the original sense, using the map $\Kf\to\Kfs$ that conflates all formulae related by supercommutation \eqref{superc}. 
 
Supercommutativity is our version of the full spin-statistics connection.%
\footnote{\label{spinstatfn}
There are several closely related statements that can be called `the spin-statistics connection.'  In our approach, we formalize it by taking the theory-specifying differential formulae to live in the supercommutative algebra $\Kfs$. This agrees with the functional-integral approach to QFT, in which the Lagrangian density is interpreted by means of  
Grassmann-valued fields, which supercommute exactly as we have described. 

From another point of view, however, our approach may \emph{seem} to involve a false premiss. If we are to interpret the field symbols $\uPhi^\lambda_{\xi_1\cdots\xi_n}$ as fields, then, on the face of it, we seem to claim that the values of these fields at any given point commute or anti-commute. This is of course false of the operator-valued fields of QFT, where commutators (or anti-commutators) vanish only at space-like separations.
The key to resolving this apparent contradiction is to remember that one \emph{cannot} simply multiply together quantum field components at a single point: such products are not usually well defined. One must regularize these products in some way, and whatever method is used should ultimately reproduce the supercommutativity seen in the functional integral approach.   For example, in 
the interaction picture of section \ref{classical_and_quantum}, 
the interaction density is not simply a sum of products of free quantum fields and their spacetime derivatives, but, rather, the \emph{normal-ordered} counterpart of such an expression. \emph{And field operators do strictly supercommute within normal-ordered expressions.}} 
It has a natural interpretation, and independent motivation via the spin-statistics theorem, in the quantum case. In contrast, our discussion of classical field theories in section \ref{ffts} leads to purely commutative rather than supercommutative formal field theories, since the derived components of classical fields commute. 
Nonetheless, it is possible to make some sense of supercommutative classical spinorial field theories. First a trivial but important example: field theories determined by \emph{linear} dynamical equations can be modelled in this way (see Remark \ref{classicalsc} below).  
In the absence of further compelling examples, we only sketch one general approach, which mirrors the non-commutativity of quantum fields.  Suppose that $\Aa=\Aa_0\oplus\Aa_1$ is a supercommutative algebra; let $\KK=\Cinf(M,\Aa_0\otimes_\RRR V_0\oplus \Aa_1\otimes V_1)$. Then the derived components of any $\Phi\in\KK$ are functions with values in $\Aa^\CCC$,  and thus supercommute.

\subsection*{Invariance}
We arrive at our main results:
\begin{thm}[Strong Reflection Invariance]\label{thm_SR_spinor}   
If a supercommutative formal field theory is invariant under $\RbF(\Lpot)$, then it is invariant under $S \circ \RPbPF (\Lpat)$.
\end{thm} 

Strong reflection invariance entails PT and CPT theorems, by the same arguments as in section \ref{SR_tensor}. To spell things out,  we consider, as in section \ref{SR_tensor}, an arbitrary complex-linear or anti-linear involution $\$$ of $W$, and we extend this to an automorphism $C_\$$ and an anti-automorphism  $\dagger_\$$ of $\Kf$. Then it is easy to deduce

\begin{thm}[General PT/CPT Theorem]\label{thm_cpt_spinor}   
Suppose that a supercommutative formal field theory is invariant under  $\RbF(\Lpot)$. Then it is invariant under $C_\$ \circ \RPbPF (\Lpat)$ if and only if it is  $\$$-Hermitian.
\end{thm}

Note that Theorems \ref{thm_SR_spinor} and \ref{thm_cpt_spinor} subsume Theorems \ref{thm_SR_tensor} and \ref{thm_cpt_tensor}, which correspond to the the special case $V=V_0$.
The proof of Theorem \ref{thm_SR_spinor} is in Appendix \ref{proofs}. The deduction of Theorem \ref{thm_cpt_spinor} from Theorem \ref{thm_SR_spinor} is completely parallel to the deduction of Theorem \ref{thm_cpt_tensor} from Theorem \ref{thm_SR_tensor}.

\subsection*{PT and CPT} The same argument as in section \ref{pt_tensor} shows that if $\RbF(\Lpot)$ is charge-preserving, then $\RPbPF(\Lpat)$ is too.  To spell it out: since each sector $W^\epsilon\subset\Kfs$ is assumed $\RbF(\Lpot)$-invariant, it is $S\circ\RPbPF(\Lpat)$-invariant, by Theorem \ref{thm_SR_spinor}; but it is obviously $S$ invariant, so it must be $\RPbPF(\Lpat)$-invariant. Thus $\RPbPF(\Lpat)$ is charge-preserving, as claimed.

   As in section \ref{SR_tensor}, a quantum CPT theorem is recovered from Theorem \ref{thm_cpt_spinor} by setting $\$=*$; for $\$=\id, \#, *\#$ we obtain, respectively, a classical PT theorem, a classical CPT theorem, and a quantum PT theorem. The quantum CPT theorem is the most important of these: it is \emph{the} CPT theorem of Lagrangian QFT, in which its premisses (supercommutativity and $*$-Hermiticity) are widely accepted. 

\begin{example}
Again the particular spin-statistics connection that we have assumed is indeed required for Theorem \ref{thm_cpt_spinor}. Suppose instead we assumed that spinors \emph{commute} with one another. Consider the equation
\be
\bar\psi\psi  = 1, \label{eg_spinor}
\ee
where $\psi$ is a Dirac spinor field (cf. Example \ref{egDirac}).
The total reflection in $\Lp$ corresponds to two elements of $\Lpat$, which act on $\psi$ by $\psi\mapsto\pm i\gamma^5\psi$ under $\RPbPF$.
 But if spinors commute then 
under the CPT transformation $\psi \mapsto (i\gamma^5 \psi)^*$ we have $\bar\psi\psi\mapsto - \bar\psi\psi$ (cf. the appeal to fermion anti-commutation in equation (3.147) of \cite{psQFT}).
Hence, \eqref{eg_spinor} transforms to $-\bar\psi\psi=1$, which 
is actually incompatible with \eqref{eg_spinor}.
\end{example}

\begin{remnum}\label{classicalsc}\label{freesym3}The classical PT theorem of section \ref{SR_tensor} applied only to commutative tensor fields, for which the requirement of $\id$-Hermiticity is trivial. It is  no longer trivial for spinor fields, although it holds for a wider class than merely tensor fields. For example, suppose that a classical spinorial field  theory $\DD$ is specified by \emph{linear} differential formulae, like the free Dirac equation.  The span $\Df\subset\Kf_s$ of those linear formulae is an $\id$-Hermitian,  supercommutative formal field theory, and its classical spacetime symmetries correspond exactly to spacetime symmetries of $\DD$. Thus if $\DD$ is $\Lpo$-invariant, so is $\Df$, and our present classical PT theorem predicts PT invariance.  Note that in this case, we have made $\Df$ supercommutative in order to apply the theorem, but this supercommutativity is irrelevant to the interpretation of $\Df$ as a classical field theory: for linear equations, there is no substantial question of commutativity or supercommutativity, since there are no products to commute or supercommute. 

\end{remnum}

\subsection{Symmetries of free quantum theories.}
 \label{freesym2} Following the discussion in section \ref{classical_and_quantum}, it is useful to explain separately how Theorem \ref{thm_cpt_spinor} yields symmetries of \emph{free} quantum field theories.  Recall that the free theory is specified by a quadratic Lagrangian density, giving rise to linear field equations. The Hilbert space is related by a Fock space construction to the classical theory defined by these linear equations.  As explained in \cite{walAnt}, the construction is such that classical and quantum $\Lpo$-invariance are equivalent, and classical PT invariance is equivalent to quantum CPT invariance. So we can argue as follows. If the free quantum theory is $\Lpo$-invariant, so is the classical theory. Our classical PT theorem (which applies by Remark \ref{freesym3}) then predicts classical PT invariance, which implies quantum CPT invariance. (The hypothesis that the Lagrangian density is Hermitian is implicit in this story.  For one thing, it guarantees that there are enough solutions to the classical field equations. It is also used to define the inner product on the Hilbert space.)

A similar argument establishes that a free QFT is PT invariant if the free Lagrangian density is $*\#$-Hermitian. This Hermiticity implies that the system of linear field equations is $*\#$-Hermitian as well. Now, classical field equations can always be written using only real coefficients (cf. the discussion around \eqref{DiracEq}). This shows that the system of field equations is $C_*$-invariant, hence $\#$-Hermitian, hence, by Theorem \ref{thm_cpt_spinor}, classically CPT invariant. And this implies that the free QFT is PT invariant.

\section{Other Spacetimes, Other Groups}\label{axioms}

Our theorems apply in principle to other spacetimes besides Minkowski space, and to other groups besides the Lorentz group. Any group $\Lp=\Lpo\cup\Lpa$ will satisfy `tensorial' invariance theorems like Theorems \ref{thm_pt_tensor}--\ref{thm_cpt_tensor} as long as it satisfies conditions (PT-1)--(PT-3) of section \ref{pt_tensor}. We will obtain  `spinorial' invariance theorems like Theorems \ref{thm_SR_spinor}--\ref{thm_cpt_spinor} if $\Lp$ also satisfies (PT-4) of section \ref{covers} and (PT-5) of section \ref{holomorphic}.

In this section we consider one setting to which our results extend, and several to which they don't. Of course, we have not shown that properties (PT-1)--(PT-3) are \emph{necessary} for any invariance theorem to hold.  However, in cases where these axioms fail, we have no general grounds to \emph{expect} invariance theorems, and understanding how they fail provides a conceptual complement to ad hoc counterexamples.

\subsection*{General Non-Euclidean Signatures} 
Let $M$ be a vector space with any non-Euclidean inner product (i.e. at least one `timelike' and one `spacelike' dimension), and $\dim M\geq 3$. Let 
$\Lp$ be the group of linear isometries of $M$ with determinant $+1$. Then $\Lp=\Lpo\cup\Lpa$ has two components, and satisfies (PT-1)--(PT-5); our constructions and Theorems \ref{thm_pt_tensor}--\ref{thm_cpt_spinor} work without change. However, there are two conceptual points to make. First, $\Lpa$ consists of transformations that reverse parity and the \emph{orientation} of time. We must understand `time-reversing' in Definition \ref{f_defnq} in this sense. Note that it is possible for some timelike vectors to be reversed but others fixed: for example, with two timelike and one spacelike direction, $(t_1,t_2,x)\mapsto(t_1,-t_2,-x)$ defines an element of $\Lpa$ that fixes the timelike vector $(1,0,0)$. 
The second point is that the preimage $\Lpot$ of $\Lpo$ in $\Lp(\CCC)^\wedge$ may not be the universal cover of $\Lpo$. However, as in Remark \ref{dim3}, there is no loss of generality in using $\Lpot$ instead of $(\Lpo)^\wedge$.

\subsection*{No CPT Theorem in Dimension 2}  Let $M$ be two-dimensional Minkowski space. Then $\Lp=\Lpo\cup\Lpa$ satisfies (PT-1) but not (PT-2) or (PT-3). In fact, if $g\in\Lpoc$ is fixed by complex conjugation, then $g\in\Lpo$.  This is closely related to the failure of  (PT-4): the universal cover of $\Lpoc$ is an infinite cover. Thus we do not expect any PT or CPT theorem;  the following example shows that there \emph{cannot} be one. 

\begin{example} 
Choose a null-vector $\xi\in M$.  Then for each $g\in\Lpo$ there is some $j(g)\in\RRR$ such that $g\cdot\xi=e^{j(g)}\xi$. This gives an isomorphism $j\colon\Lpo\to\RRR$.  Let $(\R,\Lpo,V)$ be the one-dimensional representation of $\Lpo$ given by $V=\RRR$, 
$\R(g)v=e^{j(g)/4}v$.  Consider the classical field theory given by the single 
dynamical equation
$$\Phi^3\partial_\xi\Phi=1.$$ 
It is easy to check that this equation is $\Lpo$-invariant. 
However, if there is a PT (or CPT) transformation, it must act on $V$ by some scalar $\alpha\in\RRR$, and transform $\partial_\xi\mapsto-\partial_\xi$. Then it transforms our dynamical equation by 
$\Phi^3 \partial_\xi\Phi \mapsto -\alpha^4\Phi^3\partial_{\xi}\Phi.$
Our theory will only be PT invariant if $\alpha^4=-1$; this is impossible since $\alpha$ is real. Nor can we salvage the situation by the use of anti-commuting fields.
\end{example}

\subsection*{No P or T Theorems}
If in Minkowski space we consider $\Lpo\cup\Lio$ or $\Lpo\cup\Lia$,
then  condition (PT-2) fails to hold: the only conjugation-fixed elements of $\Lpoc$ are elements of $\Lp$.  Thus we do not expect invariance under (C)P or (C)T transformations, and it is well known that such symmetries fail to exist in some physical examples.

\subsection*{No Galilean Theorems} Let $M$ be a Galilean spacetime of dimension $d$, with a fixed origin. In detail, the data is a quadruple $(M,M_s,\eta_s,\eta_t)$, where $M$ is a vector space, $M_s\subset M$ is a  hyperplane, $\eta_s$ is a Euclidean metric on `space' $M_s$, and $\eta_t$ is a Euclidean metric on `time' $M/M_s$. As in Minkowski space, the automorphism group $L$ of $(M,M_s,\eta_s,\eta_t)$ has four components, $\Lpo\cup\Lio\cup\Lia\cup\Lpa$. In this case, the only conjugation-fixed elements of $\Lpoc$ lie in $\Lpo$ itself. Thus we do not expect general (C)PT, (C)T, or (C)P theorems in Galilean spacetime; the following example shows that there cannot be one.
\begin{example} 
Consider fields taking values in $V=\RRR$, with the trivial action of $\Lpo$. Choose a basis $\xi_0,\xi_1,\ldots,\xi_{d-1}$ for $M$ with $\xi_1,\ldots\xi_{d-1}\in M_s$. Then the classical field theory given by the differential equations
$$\partial_{\xi_0}\Phi=\Phi\qquad\partial_{\xi_1}\Phi=\cdots=\partial_{\xi_{d-1}}\Phi=0$$
is $\Lpo$-invariant but \emph{cannot} be invariant under any transformation that reverses time. In particular, this rules out PT and CPT theorems.
\end{example}

\section{Conclusions}
\label{conclusion}

The CPT theorem follows from the following assumptions:
\ben
\item The theory is specified by polynomials in the fields and their spacetime derivatives, i.e. by a formal field theory $\Df$.
\item $\Df$ is invariant under a classical/quantum action of the proper orthochronous Lorentz group, or of a covering group (the classical and quantum actions being equivalent for the orthochronous group). 
\item The spin-statistics connection holds (in that $\Df$ is supercommutative).
\item $\Df$ is invariant under Hermitian conjugation. 
\een
These four assumptions lead to the CPT invariance of $\Df$; there is a further interpretive assumption that this implies the CPT invariance of the underlying (non-formal) theory.

We are not aware of any way of substantially weakening these assumptions while retaining the CPT result. There are two partial exceptions. 
First, it is possible to treat some non-polynomial theories (see Example \ref{Sine-Gordon}).
Second, one can replace the proper Lorentz group by any group satisfying our axioms (PT-1)--(PT-5) (see section \ref{axioms}).%

The proof of the CPT theorem given in the present paper differs from those in Lagrangian QFT texts primarily in its completeness and its reliance on basic geometric facts.  The usual approach works only in $3+1$ spacetime dimensions, and typically restricts consideration to scalars, vectors, and Dirac spinors.  It introduces apparently \emph{ad hoc} C, P, and T transformations in these cases, and establishes CPT invariance by computing its effects on an exhaustive list of sixteen `Dirac bilinears.' These calculations are then (if ever) generalised by appealing to  a detailed classification of representations of the four-dimensional Lorentz group. In contrast, we use the general notion of \emph{complexification} 
to define and study canonical CPT transformations, at no point appealing to exhaustive classifications. As a consequence, our proof applies at once to representations of higher spin, and to Weyl as well as Dirac spinors for the case of spin $\frac{1}{2}$. It also applies in any spacetime that shares the essential group-theoretic facts: it works for Minkowski space  and, with only minor modifications, for any non-Euclidean signature in dimension at least 3.   Our approach also differs from the axiomatic treatments of e.g. Streater and Wightman in that it applies directly to the the kind of Lagrangian quantum field theories that enjoy empirical success.

Our proof of the quantum CPT theorem also establishes quantum PT, classical PT and classical CPT theorems. All are equally theorems; the reason for regarding the quantum CPT result as the most interesting of the four is (`only') that the required supercommutativity assumption has independent motivation in the quantum but not in the classical case, while the required Hermiticity assumption has independent motivation in the quantum CPT but not in the quantum PT case.

\appendix

\section{Mathematical Background}
\label{maths}

\def\header#1{\bigskip\noindent{\textsc{\large #1}}}

In this appendix we recall the mathematical background used in the main text. The aim is to point out the important ideas, rather than to provide a detailed exposition. In doing so, we fix some useful terminology and notation.

\header{I. Representations and Algebras}
\subsection{Vector Spaces}\label{vector_spaces}

All vector spaces we consider will be \emph{real},  and all maps between vector spaces are \emph{real-linear}. We also consider complex vector spaces, which may be considered as real ones with additional structure. If $V,V'$ are complex vector spaces, then a (real-linear) map $f\colon V\to V'$ is \emph{complex-linear} if $f(iv)=if(v)$ for all $v\in V$, and \emph{anti-linear} if $f(iv)=-if(v)$ for al $v\in V$.

A \emph{complex subspace} of a complex vector space $V$ is a subset closed under addition and multiplication by complex numbers. In contrast, 
a \emph{complex affine subspace} $W\subset V$ is a subset such that if $x,y\in W$ and $a\in\CCC$ then $ax+(1-a)y\in W$. A complex affine subspace is a complex subspace if and only if it contains $0\in V$. 
\label{app:affine}

\subsection{Lie groups} \label{Lie_groups} 
A  \emph{Lie group} is a  manifold $G$ that is also a group, for which the multiplication and inversion maps are smooth. \emph{Homomorphisms} between Lie groups are required to be smooth.  

\begin{example} For any finite-dimensional vector space $V$, the group $\GL(V)$ of (real-linear) automorphisms of $V$ is naturally a Lie group. It has two connected components. The component $\GL_+(V)$ containing the identity consists of those linear maps with positive determinants. 
\end{example}

A Lie group is `complex' if it is in fact a holomorphic manifold, and the multiplication and inversion maps are holomorphic. A homomorphism between complex Lie groups may or may not be holomorphic in the usual sense for maps between holomorphic manifolds. 

\begin{example} If $W$ is a finite-dimensional complex vector space, then the group $\GL_\CCC(W)$ of complex-linear automorphisms is a complex Lie group. 
\end{example}

\medskip
\subsubsection{\sc Covering Groups} \label{covering_groups}
A homomorphism $G'\to G$ of Lie groups is a \emph{covering map} if it is surjective and a local diffeomorphism; we say that $G'$ is a \emph{covering group} or simply a \emph{cover} of $G$.
A covering $\pi\colon G'\to G$ is \emph{trivial} if there is a section, i.e. a homomorphism $s\colon G\to G'$ such that $\pi\circ s$ is the identity map. If $G'$ is connected, then $\pi$ is trivial if and only if it is an isomorphism.

Covering groups have a \emph{path lifting property}. The case we need is that any homomorphism $\RRR\to G$ of Lie groups lifts uniquely to a homomorphism $\RRR\to G'$.
 
If $G$ is connected, then it has a unique \emph{universal cover} $G^\wedge$. 
The covering map $\pi_G\colon G^\wedge\to G$ is uniquely characterised by the following property: for any covering map $\alpha\colon G'\to G$, with $G'$ connected, there exists a unique covering map $\pi_{G'}\colon G^\wedge\to G'$ such that $\pi_G=\alpha\circ\pi_{G'}$. It follows that this $\pi_{G'}$ is the universal covering map for $G'$. It also follows that if $\pi\colon G'\to G^\wedge$ is a covering map, and $G'$ is connected, then $\pi$ is an isomorphism. (Indeed, $\pi_G\circ\pi\colon G'\to G$ is a covering map, so we obtain a covering map $\pi_{G'}\colon G^\wedge\to G'$, which must be inverse to $\pi$.)
 
\medskip
\subsubsection{\sc Lie Algebras}\label{Lie_algebras}  Given a Lie group $G$, let $\Lie(G)$ denote the tangent space to $G$ at the identity. If $G$ is a complex Lie group, then $\Lie(G)$ is a complex vector space. A homomorphism $\alpha\colon G\to H$  induces a linear map $\Lie(\alpha)\colon \Lie(G)\to\Lie(H)$ (the derivative of $\alpha$ at the identity);  if $G$ is connected, $\alpha$ is completely determined by $\Lie(\alpha)$. 
If $G$ and $H$ are connected, then $\alpha$ is a covering map if and only if $\Lie(\alpha)$ is an isomorphism. 

\begin{example} For any finite-dimensional vector space $V$,  $\Lie(\GL(V))$ is the vector space of all linear maps $f\colon V\to V$. For $\Lpo\subset\GL(M)$, one can show that $\Lie(G)\subset\Lie(\GL(M))$ consists of those $f$ such that
$$\eta(f(v),w)+\eta(v,f(w))=0\mbox{ for all $v,w\in M$}.$$
If $V$ is complex, then $\Lie(\GL_\CCC(V))$ is the space of all \emph{complex}-linear maps $V\to V$. 
\end{example}

\subsection{Representations}\label{representations}  
All representations in this paper are either \emph{finite dimensional} or what we call \emph{quasi-finite}.
A finite-dimensional representation of a Lie group is a triple $(\rho,G,V)$, where  $V$ is a finite-dimensional vector space, $G$ is a Lie group, and $\rho\colon G\to \GL(V)$ is a homomorphism of Lie groups. Often we speak of `the representation (or action) $\rho$ of $G$ on $V$.'
A representation is \emph{holomorphic} if $V$ is a complex vector space, $G$ a complex Lie group, and $\rho$ is a holomorphic homomorphism.

\subsubsection{Quasi-finite Representations}\label{quasifinite}
We also consider certain infinite-di\-men\-sion\-al representations $(\rho,G,V)$. Again $V$ is a vector space and $\rho$ is a homomorphism  $G\to\GL(V)$.  But there is a technicality here in that if $V$ is infinite dimensional, then $\GL(V)$ is not a Lie group in the usual sense. 
However, all representations that we consider are \emph{quasi-finite} in the following sense: there is a sequence of finite-dimensional subspaces
$V^1\subset V^2\subset\cdots$
such that the union is $V$,  each $V^n$ is $\rho(G)$-invariant, and the resulting homomorphisms $\rho^n\colon G\to\GL(V^n)$ are smooth. 
These  conditions allow us to apply notions that are initially valid only for finite dimensional representations. For example, a quasi-finite representation is `holomorphic' if every $\rho^n$ is holomorphic. We will use this technical device without further comment.

\subsection{Algebras}\label{algebras}
An \emph{algebra} is a real vector space $\Aa$ with a bilinear, associative, unital multiplication $\Aa\otimes_\RRR\Aa\to \Aa$. 
A \emph{homomorphism} of algebras is a linear map preserving the multiplication and the unit.  An algebra is \emph{complex} if $\Aa$ is a complex vector space, and the multiplication is complex-bilinear.  A (real-linear) homomorphism $f$ between complex algebras may be complex-linear or anti-linear.

\subsection{Graded Vector Spaces and Algebras}\label{gradings}
A  \emph{vector space} $V$ is $\ZZZ$-graded (or merely \emph{graded}) if it is given as a direct sum $$V=\bigoplus_{n\in \ZZZ} V_n$$
of vector spaces $V_n$ labelled by integers. 
\label{graded_algebras} An \emph{algebra} is graded if it is a graded vector-space and  
$ab\in \Aa_{m+n}$ for all $a\in\Aa_m,b\in\Aa_n$.

\header{II. Free Algebras}

\subsection{Tensor Algebras}\label{tensor_algebras} Suppose that $V$ is a real vector space. The \emph{tensor} (or \emph{free}) algebra $TV$ consists of formal sums of products of elements of $V$. More precisely, 
 if $v_1,\ldots,v_n$ is a basis for $V$, then $TV$ consists of \emph{non-commutative} polynomials in variables $v_1,\ldots,v_n$, with real coefficients.  More conceptually, 
$$TV=\RRR\oplus V\oplus(V\otimes_\RRR V)\oplus(V\otimes_\RRR V\otimes_\RRR V)\oplus\cdots=\bigoplus_{n\geq0}V^{\otimes n}.$$
The key property of $TV$ is that any real-linear map $V\to V$ extends uniquely to a homomorphism $TV\to TV$ of algebras.

\subsection{Free Complex Algebras}\label{free_algebras}
The analogous construction for a complex vector space $V$ 
is the \emph{free complex algebra} $\FF(V)$ defined by 
$$\FF(V)=\CCC\oplus V\oplus(V\otimes_\CCC V)\oplus(V\otimes_\CCC V\otimes_\CCC V)\oplus\cdots=\bigoplus_{n\geq0}V^{\otimes n}.$$
If $v_1,\ldots,v_n$ is a complex basis for $V$, then $\FF(V)$ consists of non-commutative polynomials in variables $v_1,\ldots,v_n$, with complex coefficients.   The key property is that any complex-linear or anti-linear map $V\to V$ extends uniquely to a complex-linear or anti-linear homomorphism $\FF(V)\to\FF(V)$ of algebras. 

\subsection{Commutativity and Supercommutativity}\label{variations}\label{commutativity}
The \emph{free commutative complex algebra} $\FF_c(V)$ is constructed from $\FF(V)$ by imposing the relation
$$ab=ba\in\FF_c(V)\qquad\mbox{for all $a,b\in V.$}$$
Thus elements of $\FF(V)$ that differ only by commutation are considered equal as elements of $\FF_c(V)$. 
Concretely, if $v_1,\ldots,v_n$ is a basis for $V$, then $\FF_c(V)$ is the algebra of (usual, commutative) polynomials in variables $v_1,\ldots,v_n$, with complex coefficients.  

\label{supercommutativity} 
If $V$ is a graded complex vector space, then the 
\emph{free supercommutative complex algebra} $\Fs(V)$ is constructed from $\FF(V)$ by imposing the `supercommutativity' relations
$$ab=(-1)^{mn}ba\quad\mbox{for all $a\in V_m$, $b\in V_n$.}$$
For example, if $V=V_0$ then $\Fs(V)=\FF_c(V)$, and if $V=V_1$ then $\Fs(V)$ is known as the \emph{Grassmann} or \emph{exterior} algebra of $V$.
If $V$ is a graded vector space, then $\Fs(V)$ (like $TV$, $\FF(V)$, or $\FF_c(V)$) is naturally a graded algebra.
 Concretely, $(\Fs(V))_m$ is spanned by products $v_{d_1}\cdots v_{d_n}$, with each $v_{d_k}\in V_{d_k}$ and $\sum d_k=m$.

\subsection{Extending Representations}\label{extending_reps}
As we have said, any complex-linear or anti-linear map $V\to V$ extends uniquely to a homomorphism 
$\FF(V)\to\FF(V)$ of algebras.
In particular, suppose that $(\rho,G,V)$ is a representation such that each transformation $\rho(g)$ is either complex-linear or anti-linear; then there is a unique representation $(\rho^\FF,G, \FF(V))$ such that each $\rho^\FF(g)$ is an algebra automorphism.  Explicitly: $\rho^\FF(g)\colon \FF(V)\to\FF(V)$ is the unique homomorphism of algebras extending $\rho(g)\colon V\to V$.
One similarly obtains representations
 $(\rho^{\Fc},G,\Fc(V))$ and 
 $(\rho^{\Fs},G,\Fs(V))$.

\header{III. Complexification}

\subsection{Vector Spaces}\label{complexify_vector_spaces}
\emph{Complexification} is a universal way to replace real objects by complex ones.  For a real vector space $V$, its complexification $V^\CCC$ is  
the tensor product
$$V^\CCC=\CCC\otimes_\RRR V:=\{a+bi\,\mid\,a,b\in V\}.$$ 
It is characterised by the property that any real-linear map $\alpha\colon V\to V'$ into a complex vector space $V'$ extends uniquely to a complex-linear map $\alpha\Hol\colon V^\CCC\to V'$. In other words, the set of real-linear maps $V\to V'$ 
is naturally identified with the set of complex-linear maps $V^\CCC\to V'$. 

\begin{rem} If $V$ is complex, then we can still form its complexification $V^\CCC$. This is not equal to $V$, unless of course $V=0$. In general, if $V$ has $n$ real dimensions, then $V^\CCC$ has $2n$ real dimensions.
\end{rem}

\subsection{Lie Groups}\label{complexify_Lie_groups}
\label{complexifying_lie_groups}
Given a \emph{connected} Lie group $G$, its complexification $G^\CCC$ is a complex Lie group, equipped with a homomorphism $i\colon G\to G^\CCC$, characterised by the following property: any homomorphism $\alpha\colon G\to G'$ into a complex Lie group $G'$ extends uniquely to a holomorphic homomorphism $\alpha\Hol\colon G^\CCC\to G'$. That is, $\alpha\Hol$ uniquely satisfies $\alpha\Hol\circ i = \alpha$.

For the construction of $G^\CCC$, see \cite[Ch. 3, \S6, Prop. 20]{bouGAL}. The essential fact for us is that $\Lie(G^\CCC)=(\Lie G)^\CCC$
(the right-hand side being the complexification in the sense of vector spaces).

\begin{rem}Note that, in general, $G$ may not be a subgroup of $G^\CCC$; there is only a homomorphism between them.  However, this homomorphism happens to be injective in most cases of interest, and in particular for the Lorentz group.\end{rem}

\begin{example}  If $V$ is a finite-dimensional real vector space, then $\GL_+(V)$ is a connected Lie group. If $\dim V>1$ then $\GL_+(V)^\CCC = \GL_\CCC(V^\CCC)$. However, if $\dim V=1$ then $\GL_+(V)$ is isomorphic to the additive group of real numbers,  and $\GL_+(V)^\CCC$ to the additive group of complex numbers, whereas $\GL_\CCC(V^\CCC)$ is isomorphic to the multiplicative group of non-zero complex numbers.
\end{example}

\subsection{Extending Representations}\label{complexify_representations}
\label{complexifying_representations}
Consider a representation $(\rho,G,V)$ of a connected Lie group $G$.
If $V$ is complex, and every $\rho(g)$ is complex-linear, then $(\rho,G,V)$ extends uniquely to a \emph{holomorphic} representation $(\rho\Hol,G^\CCC,V)$. Namely, $\rho\Hol\colon G^\CCC\to\GL_\CCC(V)$ is the unique holomorphic homomorphism extending $\rho\colon G\to\GL_\CCC(V)$.
If $V$ is not assumed to be complex, $(\rho,G,V)$ extends uniquely to a holomorphic representation $(\rho^\CCC,V^\CCC,G^\CCC)$ of $G^\CCC$ on $V^\CCC$. That is, $\rho^\CCC\colon G^\CCC\to \GL_\CCC(V^\CCC)$ is the unique holomorphic homomorphism extending $\rho\colon G\to\GL(V)\subset\GL_\CCC(V^\CCC)$.

\begin{example} Suppose $V$ is a finite-dimensional vector space of dimension at least two, and $(\rho,\GL_+(V),V)$ is the usual representation of $\GL_+(V)$ on $V$. Then $\rho^\CCC$ is just the natural representation of $\GL_\CCC(V^\CCC)$ on $V^\CCC$. 
\end{example}

\subsection{Complex Conjugation}\label{complex_conjugation}

For a vector space $V$,  one has the standard complex conjugation on $V^\CCC$, defined by $*_V\colon a+bi\mapsto (a+bi)^*=a-bi.$  Note that such a map  is \emph{not} automatically defined on every complex vector space, but only on a complexification. 
Similarly, for a connected Lie group $G$ there is a complex conjugation 
 $*_G\colon G^\CCC\to G^\CCC$,   characterised by the equation
$\Lie(*_G)=*_{\Lie(G)}$ as maps $\Lie(G^\CCC)\to\Lie(G^\CCC).$
The essential fact for us is this:
\begin{quote} Let  $(\rho,G,V)$ be a representation of a connected Lie group. Then 
$(\rho^\CCC(g)v)^*=\rho^\CCC(g^*)v^*$ for all $g\in G^\CCC,v\in V^\CCC$.
\end{quote}

\section{Clifford Algebras and Pin Groups}\label{Clifford}

In section \ref{covers}, we described covering groups of $\Lp$ in terms of the universal cover $\Lp(\CCC)^\wedge$. Another description, more common but less suited to our needs, can be given in the language of \emph{Clifford algebras}.  Since Clifford algebras are routinely used in spinor theory,  we now explain how our covering groups appear in that framework. 

\subsection{} The \emph{Clifford algebra} $\Cliff$ is derived from the free algebra $\FF(M^\CCC)$ by imposing the relations $vw+wv=2\eta^\CCC(v,w)$ for all $v,w\in M^\CCC$.

 Any sequence $v_1,\ldots,v_n\in M^\CCC$ defines a product $v_1v_2\cdots v_n\in\Cliff$. The \emph{Pin group}\footnote{Sometimes this term refers only to what we call ${}^a\!P$.} $P$ consists of all such products in which $v_1,\ldots,v_n$ satisfy $\eta^\CCC(v_i,v_i)=\pm 1$. This is a group because each $v_i$ has a multiplicative inverse, equal to  $\pm v_i$. 
$P$ has four connected components
$$P= 
{}^a\!P_+\cup
{}^b\!P_+\cup
{}^a\!P_-\cup
{}^b\!P_-$$
 which can be described as follows. 
The identity component of $P$ is ${}^a\!P_+$. 
The subgroup $P_+= {}^a\!P_+\cup{}^b\!P_+$ consists of products $v_1\cdots v_n\in P$ with $n$ even; the subgroup ${}^a\!P= {}^a\!P_+\cup{}^a\!P_-$ consists of products $v_1\cdots v_n\in P$ with every $\eta^\CCC(v_i,v_i)=1$. 
The Pin group $P$ is a covering group of the complex Lorentz group $L(\CCC)$; the covering map $\pi\colon P\to L(\CCC)$ associates to a unit vector $v\in M^\CCC$ the reflection in the plane orthogonal to  $v$:
$$\pi(v)\colon x\mapsto 2\frac{\eta^\CCC(x,v)}{\eta^\CCC(v,v)}v -x.$$
This is a four-fold cover, with $\pi(g)=\pi(-g)=\pi(\pm ig)$ for all $g\in P$. In particular, the preimage of the identity element of $L(\CCC)$ is $\{\pm1,\pm i\}\subset\Cliff$.

\subsection{}We now describe the covering groups from \S\ref{covers} in terms of $P$. The four-fold cover $\Lptt(\CCC)$ of $\Lp(\CCC)$ can be identified with the preimage of $\Lp(\CCC)$ in $P$. Thus 
$$\Lptt(\CCC)=P_+={}^a\!P_+\sqcup{}^b\!P_+.$$
The double cover $\Lpot$ of $\Lpo$ is the preimage of $\Lpo$ in
${}^a\!P_+$. 
This extends to a double cover of $\Lp$ in two different ways: \begin{enumerate}
\item[(a)]  $\Lpti=\Lpot\sqcup\Lpati$, where $\Lpati$ is the preimage of $\Lpa$ in ${}^a\!P_+$. 
\item[(b)]  $\Lpt=\Lpot\sqcup\Lpat$, where $\Lpat$ is the preimage of $\Lpa$ in ${}^b\!P_+$. 
\end{enumerate}
The element called $\tau$ in \S\ref{covers}
is $-1\in\Cliff$, and what we called $I$ is $i\in\Cliff$.

\subsection{} \label{PT-5-alt}  We can view $\Cliff$ as the complexification of the real Clifford algebra $\Cliff_\RRR$, constructed from the tensor algebra $TM$ by the relations $vw+wv=2\eta(v,w)$. The subgroup of $P_+$ fixed by complex conjugation is exactly the cover $\Lpat$. This is a direct analogue of (PT-3) and can be used to give an alternative proof of (PT-5).

\section{Proofs}
\label{proofs}

\subsection{Proof of Lemma \ref{complexlorentz}}

Let $j$ be the inclusion of $\Lpo$ into $\Lp(\CCC)$. It extends to a holomorphic map $j\Hol\colon \Lpoc\to\Lp(\CCC)$. The claim is that $j\Hol$ is an isomorphism. 

First we prove that the corresponding map $\Lie(j^\CCC)$ of Lie algebras is an isomorphism, so that $j\Hol$ is a covering map (\ref{Lie_algebras}). 
This is perhaps clearest in coordinates. Choose an orthonormal basis of $M$, so that $M=\RRR^d$ and the metric is represented by a diagonal matrix $\eta=\diag(1,-1,\ldots,-1)$. Standard considerations identify $\Lie(\Lp(\CCC))$ with a set of matrices
$$
\Lie(\Lp(\CCC))=\{f\in\Mat(d,\CCC)\,\mid\,f\eta=-\eta f^T\}$$
and $\Lie(\Lpo)\subset\Lie(\Lp(\CCC))$ with the subspace of \emph{real} such matrices.  The claim is just that for any real matrices $f_1$ and $f_2$, $f=f_1+if_2$ lies in $\Lie(\Lp(\CCC))$ if and only if $f_1$ and $f_2$ lie in $\Lie(\Lpo)$. This is easy to check.

By \ref{covering_groups} we therefore have covering maps
$$\Lp(\CCC)^\wedge\overset{\pi}\too\Lpoc\overset{j\Hol}\too\Lp(\CCC)$$
whose composition is the universal covering map $\pi_{\Lp(\CCC)}\colon\Lp(\CCC)^\wedge\to\Lp(\CCC)$.
Now we appeal to two well-known topological facts. First, $\pi_{\Lp(\CCC)}$ is two-to-one.  
 Therefore either $\pi$ or $j\Hol$ must be an isomorphism. Second, the preimage of $\Lpo$ in $\Lp(\CCC)^\wedge$ is a non-trivial double cover of $\Lpo$ (in fact  it is the universal cover if $\dim M>3$). In particular,  
there can be no map $i\colon \Lpo\to\Lp(\CCC)^\wedge$ such that $\pi_{\Lp(\CCC)}\circ i=j$. But by definition of $j\Hol$, there \emph{is} a map $i\colon \Lpo\to\Lpoc$ such that $j\Hol\circ i=j$. Therefore $j\Hol$ rather than $\pi$ must be the isomorphism.

\subsection{Proof of Lemma \ref{lem2}}

The claim is that for any $g\in\Lp$ and $v\in V$, $\RC(g)v$ is in $V$. This is equivalent to saying that it is fixed by the complex conjugation on $V^\CCC$: 
$$\left( \RC(g)v\right)^* = \RC(g)v.$$
Since $\RC$ is a complexification, $(\RC(g)v)^* = \RC(g^*)v^*$ (\S\ref{complex_conjugation}), and by assumption $v^*=v$. The lemma therefore follows from the property  $g^*=g$ stated as (PT-3) and established by Example \ref{eg_vector}.

\subsection{Proof of Lemma \ref{prop2}}
This is just a special case of the following:
\begin{lemma} \label{lem4}
	Let $(R,G, A)$ be a  complex-linear representation of a connected Lie group $G$ on a complex vector space $A$, and $(R\Hol,G^\CCC,A)$ its holomorphic extension.  Any complex affine subspace $A' \subset A$ invariant under $R(G)$ is also invariant under $R\Hol(G^\CCC)$.
	\end{lemma}
\begin{proof} We can assume that $A$ is finite-dimensional. Let $H\subset\GL_\CCC(A)$ be the group of all complex-linear maps preserving $A'\subset A$. This $H$ is a holomorphic subgroup of $\GL_\CCC(A)$. The representation $R$ factors through a homomorphism 
$R_1\colon G\to H$. This extends to a holomorphic map $R_1\Hol\colon G^\CCC\to H\subset\GL_\CCC(A)$. By the uniqueness of $R\Hol$, we must in fact have $R_1\Hol=R\Hol$, so $R\Hol(G^\CCC)$ preserves $A'$.
\end{proof}

\subsection{Proof of Lemma \ref{prop3}}

For each $g\in\Lp$, $\RbF\Hol(g)$ is a complex-linear automorphism of $\Kf$, as is $\RPbPF(g)$; to show that they are equal, it suffices to show that they transform the space $W\otimes_\RRR TM\subset\Kf$ of field symbols in the same way.  Consider the representation $R$ of $\Lpoc$ on $W\otimes_\RRR TM$ defined by
\be\label{Arr!}R(g)(\lambda\otimes(\xi_1\cdots\xi_n))= (\lambda\circ\rho^\CCC(g)\inv)\otimes(\omega^\CCC(g)\xi_1\cdots\omega^\CCC(g)\xi_n).\ee
Some explanation is needed for this formula to make sense. First, we have defined $W$ to be the space of real-linear maps $V\to\CCC$.  According to \S\ref{complexify_vector_spaces},  we could equivalently define $W$ to be the space of complex-linear  maps $\lambda\colon V^\CCC\to\CCC$.  With this in mind, the expression
$\lambda\circ\rho^\CCC(g)\inv$ defines an element of $W$.
Second, we consider the $\xi_i$ as \emph{complex} vectors, i.e. elements of $M^\CCC$, so that $\xi_1\cdots\xi_n$ and $\omega^\CCC(g)\xi_1\cdots\omega^\CCC(g)\xi_n$ both lie in the complex tensor algebra $\FF(M^\CCC)$. This is legitimate because any complex coefficients can be absorbed into $W$: formally, we have
$$W\otimes_\RRR TM=W\otimes_\CCC\CCC\otimes_\RRR TM=W\otimes_\CCC\FF(M^\CCC).$$
Thus the right-hand side of \eqref{Arr!} makes sense as an element of 
$W\otimes_\RRR TM$.  

Now $R$ is a holomorphic representation of $\Lpoc$ on $W\otimes_\RRR TM$, and its restriction to $\Lpo\subset\Lpoc$ is just $\RbF$. This property uniquely characterises $\RbF\Hol$, so $R=\RbF\Hol$.  On the other hand, for $g\in\Lp$, \eqref{Arr!} is the very definition of $\RPbPF(g)$. Therefore $\RbF\Hol=\RPbPF$ as representations of $\Lp$.

\subsection{Proof of Lemma \ref{complexlorentzcover}}
The inclusion $j\colon \Lpot\to\Lp(\CCC)^\wedge$ extends to a holomorphic homomorphism $j\Hol\colon \Lpotc\to\Lp(\CCC)^\wedge$.  
Since $\Lpot$ is a cover of $\Lpo$, and $\Lp(\CCC)^\wedge$ is a cover of $\Lp(\CCC)$, they have the same Lie algebras as $\Lpo$ and $\Lp(\CCC)$ respectively; the map $\Lie(j\Hol)$ of Lie algebras is \emph{exactly} the one appearing in the proof of Lemma \ref{complexlorentz}.
The argument there shows that $\Lpotc$ must be a covering space of $\Lp(\CCC)^\wedge$. But $\Lp(\CCC)^\wedge$ is the \emph{universal} covering space of $\Lp(\CCC)$.  Therefore the covering map $\Lpotc\to\Lp(\CCC)^\wedge$ must be an isomorphism (see \ref{covering_groups}).

\subsection{Proof of Lemma \ref{lemRepP'}}
Momentarily assume (PT-5), that $g^*=g\tau$ for all $g\in\Lpati$. Suppose given $v\in V_n$. It's enough to show is that $i^n\RC(g)v$ lies in $V_n$. 

First we show that it lies in $V$, i.e. that it is fixed by complex conjugation (cf. the proof of Lemma \ref{lem2}).  We have
$(i^n\RC(g)v)^*=(-i)^n\RC(g^*)v^*=(-i)^n\RC(g\tau)v=
(-i)^n\RC(g)\RC(\tau)v=(-i)^n\RC(g)(-1)^nv=i^n\RC(g)v$. Thus $i^n\RC(g)v$ lies in $V$.  

Next, to show that it lies in $V_n$, it remains to note that $\RC(\tau)(i^n\RC(g)v)=i^n\RC(g)\RC(\tau)v=(-1)^n(i^n\RC(g)v)$.

Finally, to establish (PT-5), choose a unit timelike vector $e_0$ and an orthogonal unit spacelike vector $e_1$. Let $R(\theta)\in\Lp(\CCC)$ be the transformation defined by rotating $M^\CCC$ by $\theta$ in the plane spanned by $ie_0$ and $e_1$. That is,
$$R(\theta)ie_0=(\cos\theta)ie_0-(\sin\theta)e_1\qquad
R(\theta)e_1=(\sin\theta)ie_0+(\cos\theta)e_1$$
and $R(\theta)v=v$ if $v$ is orthogonal to $e_0$ and $e_1$.
Thus $R$ is a homomorphism $\RRR\to\Lp(\CCC)$ of Lie groups,  and 
(\S\ref{covering_groups}) there is a unique way to lift $R$ to a homomorphism $R^\wedge\colon\RRR\to\Lp(\CCC)^\wedge$. 
$R([0,2\pi])$ is a loop in $\Lp(\CCC)$, starting and finishing at the identity. This loop is well known to be \emph{homotopically non-trivial}; that just means that $R^\wedge(2\pi)=\tau$.
It is also easy to check that $R(\theta)^*=R(\theta)\inv$, and the uniqueness of the lifting then implies $R^\wedge(\theta)^*=R^\wedge(\theta)\inv$.

Set $h=R^\wedge(\pi)$. Then $h\in\Lpati$, $h^2=R^\wedge(2\pi)=\tau$, and $h^*=h\inv$. Therefore $h^*=h\tau$.
It follows that $g^*=g\tau$ for \emph{any} $g\in\Lpati$: there exists $g_0\in\Lpot$ such that $g=g_0h$, and then $g^*=g_0^*h^*=g_0h\tau=g\tau$.

\subsection{Proof of Theorem \ref{thm_SR_spinor}} We proceed in three steps.

\subsubsection*{Step 1.}
Since the formal field theory $\Df$ is invariant under $\RbFs(\Lpot)$, it follows (via Lemma \ref{lem4}) that $\Df$ is invariant under 
$\RbFs\Hol(\Lpti)$. This is the analogue of Lemma \ref{prop2}.

\subsubsection*{Step 2.} Recall now that $\Kfs$ is a graded algebra, 
 $$\Kfs=(\Kfs)_0\oplus(\Kfs)_1\oplus(\Kfs)_2\oplus\cdots$$ 
 with $(\Kfs)_m$ spanned by monomials containing $m$ fermionic field symbols -- that is, $m$ factors from $W_1 \otimes_\RRR TM$.  We next establish that
\be\label{ii}S\circ\RPbPF(g)(X)=i^{m^2}\RbF\Hol(I\inv g)(X)\ee
for $X\in(\Kfs)_m$ and $g\in\Lpat$.  This is an analogue of Lemma \ref{prop3}.

Formula \eqref{Arr!} and subsequent comments define a representation $R$ of $\Lpotc$ on $W\otimes_\RRR TM$, and the argument there shows that $R=\RbF\Hol$. In this case, evaluating \eqref{Arr!} at $I\inv g\in\Lpati$ shows that
$$\RbF\Hol(I\inv g)(\lambda\otimes(\xi_1\cdots\xi_n))=\begin{cases}
\phantom{i}\RPbPF(g) (\lambda\otimes(\xi_1\cdots\xi_n)) & \mbox{for $\lambda\in W_0$ }\\
i \RPbPF(g) (\lambda\otimes(\xi_1\cdots\xi_n)) &  \mbox{for $\lambda\in W_1$}.\end{cases}$$
These formulae determine the action of  $\RbF\Hol(I\inv g)$ on all of $\Kfs$:  
\be\label {uu}\RbF\Hol(I\inv g)(X)=i^m \RPbPF(g)(X)\ee
for $X\in(\Kfs)_m$. On the other hand, it is easy to check that 
\be\label{yy}S(X)=(-1)^{m(m-1)/2}X.\ee
 This is analogous to the statement that $S$ is the identity map on $\Kfc$; it is here and only here that we use supercommutativity. Equations \eqref{uu} and \eqref{yy} imply \eqref{ii}.

\subsubsection*{Step 3.}Applying Steps 1 and 2, it only remains to argue that $\Df$ is invariant under the transformation $\alpha$ defined 
on $X\in(\Kfs)_m$ by
$$\alpha(X)=i^{m^2}X=\begin{cases} \phantom{i}X &\mbox{if $m$ is even} \\
 iX & \mbox{if $m$ is odd}.\end{cases}$$
 Note that
$$\RbF(\tau)(X)=\begin{cases} \phantom{-{}}X &\mbox{if $m$ is even} \\
 -X & \mbox{if $m$ is odd}.\end{cases}$$
From this it is easy to check that, for any $F\in\Kfs$,
\be\label{affinemagic}\alpha(F)=\left(\frac{1+i}2\right)F+\left(1-\frac{1+i}2\right)\RbF(\tau)F.\ee
Now, suppose given $F\in\Df$. Since $\Df$ is $\RbF(\Lpot)$-invariant, we also have $\RbF(\tau)F\in \Df$. 
Since $\Df$ is an affine subspace,  $F\in\Df$, and $\RbF(\tau)F\in\Df$, it follows from \eqref{affinemagic} that 
$\alpha(F)$ lies in $\Df$ as well.

\section*{Acknowledgements}

The authors would like to thank Robert Geroch and David Wallace for very helpful discussions, and Merton College, Oxford, for the Junior Research Fellowships that enabled us to forge our collaboration. Part of HG's contribution to this paper was supported by Perimeter Institute for Theoretical Physics; research at Perimeter Institute is supported by the Government of Canada through Industry Canada and by the Province of Ontario through the Ministry of Economic Development and Innovation. Part of TT's contribution was supported by a Seggie Brown Research Fellowship.

\bibliographystyle{apacite}

\bibliography{hils_bib}

\end{document}